\documentclass[a4paper, 10pt]{article}
\usepackage[mathscr]{eucal}
\usepackage{amsmath,amsfonts,amssymb,amsthm}
\usepackage{times}

\advance\textheight\topskip
\numberwithin{equation}{section} \makeatletter
\@addtoreset{equation}{section}

\renewcommand{\tilde}{\widetilde}
\renewcommand{\hat}{\widehat}

\newcommand{\bref}[1]{\textbf{\ref{#1}}}
\newcommand{\gh}[1]{\mathrm{gh}(#1)}

\newcommand{\dd}{\partial}

\newcommand{\tensor}{\otimes}

\renewcommand{\geq}{\,{\geqslant}\,}
\renewcommand{\leq}{\,{\leqslant}\,}

\newcommand{\inner}[2]{\langle #1{,}\,#2\rangle}
\newcommand{\binner}[2]{%
  {\langle}\kern-4.15pt{\langle}#1{,}\,#2{\rangle}\kern-4.15pt{\rangle}}

\newcommand{\scommut}[2]{\{#1{,}\,#2\}_\ast}
\newcommand{\qcommut}[2]{[#1{,}\,#2]_*}
\newcommand{\pb}[2]{\left\{{}#1{},{}#2{}\right\}}

\newcommand{\half}{\mathchoice{%
    \ffrac{1}{2}}{\frac{1}{2}}{\frac{1}{2}}{\frac{1}{2}}}

\newcommand{\ffrac}[2]{\raisebox{.5pt}%
  {\footnotesize$\displaystyle\frac{#1}{#2}$}\kern1pt}

\newcommand{\dl}[1]{\mathchoice{\ffrac{\dd}{\dd #1}}{\frac{\dd}{\dd  #1}}{\ffrac{\dd}{\dd #1}}{\ffrac{\dd}{\dd #1}}}

\newcommand{\manifold}[1]{\mathscr{#1}}
\newcommand{\manX}{\manifold{X}}

\newcommand{\fC}{\mathbb{C}}

\usepackage[normalem]{ulem}

\usepackage{amsfonts}
\usepackage{amssymb}
\usepackage{amsmath}
\usepackage{amsthm}
\usepackage{thmtools}
\usepackage{stmaryrd}
\usepackage{amsbsy}
\usepackage{mathtools}
\usepackage[dvipsnames]{xcolor}
\definecolor{arXiv}{named}{Maroon}
\definecolor{ColorCite}{named}{BrickRed}
\definecolor{ColorLink}{named}{NavyBlue}
\definecolor{ColorURL}{named}{RoyalBlue}
\definecolor{ColorMail}{named}{NavyBlue}
\usepackage{fontawesome}
\usepackage{adforn}
\usepackage{dsfont}
\usepackage[pdftex,
  colorlinks=true,
  pdfstartview=FitV,
  linkcolor=ColorLink,
  citecolor=ColorCite,
  urlcolor=ColorURL,
  hyperindex=true,
  hyperfigures=false]
  {hyperref}
\usepackage{cite}

\usepackage{graphicx}
\usepackage{tabularx}
\usepackage{lscape}
\usepackage{fancyhdr}
\usepackage{mdframed}
\usepackage{indentfirst}
\usepackage{relsize}
\usepackage{chngpage}
\usepackage{enumitem}

\setitemize[1]{label=$\bullet$}

\counterwithin{equation}{section}

\usepackage{tikz, tikz-cd, animate}
\usetikzlibrary{decorations.pathmorphing, arrows, matrix}
\usepackage{genyoungtabtikz}
\Yboxdim{7pt}
\Ylinethick{1pt}
\newtheorem{theorem}{Theorem}[section]
\newtheorem*{theorem*}{Theorem}
\newtheorem{lemma}[theorem]{Lemma}
\newtheorem*{lemma*}{Lemma}
\newtheorem{proposition}[theorem]{Proposition}
\newtheorem{corollary}[theorem]{Corollary}

\newcommand{\pl}{\partial}
\newcommand{\Tr}{\mathrm{Tr}}
\newcommand{\PPP}{{\boldsymbol{P}}}
\newcommand{\LLL}{{\boldsymbol{L}}}
\newcommand{\KKK}{{\boldsymbol{K}}}
\newcommand{\DDD}{{\boldsymbol{D}}}
\newcommand{\TTT}{{\boldsymbol{T}}}
\newcommand{\hs}{{\ensuremath{\mathfrak{hs}}}}
\newcommand{\bloody}{D}
\newcommand{\tr}{\mathrm{tr}}
\newcommand{\C}{\mathbb{C}}
\newcommand{\N}{\mathbb{N}}
\newcommand{\R}{\mathbb{R}}

\newcommand{\cA}{\mathcal{A}}
\newcommand{\cC}{\mathcal{C}}
\newcommand{\cD}{\mathcal{D}}
\newcommand{\cL}{\mathcal{L}}

\newcommand{\cO}{\mathcal{O}}
\newcommand{\cS}{\mathcal{S}}
\newcommand{\cU}{\mathcal{U}}

\begin{document}
%%%%%%%%%%%%%%%%%%%%%%%%%%%%%%%%%%%%%%%%%%%%%%%%%%%%%%%%%%%%%
\pagenumbering{gobble}
\hfill
\vskip 0.01\textheight
\begin{center}
{\Large\bfseries 
Covariant action for conformal higher spin gravity}

\vskip 0.03\textheight
\renewcommand{\thefootnote}{\fnsymbol{footnote}}
Thomas \textsc{Basile}${}^a$, Maxim \textsc{Grigoriev}${}^{b,c}$ \& Evgeny \textsc{Skvortsov}\footnote{Research Associate of the Fund for Scientific Research -- FNRS, Belgium}${}^{a}$\footnote{Also on leave from Lebedev Institute of Physics, Moscow, Russia}
\renewcommand{\thefootnote}{\arabic{footnote}}
\vskip 0.03\textheight

{\em ${}^{a}$ Service de Physique de l'Univers, Champs et Gravitation, \\ Universit\'e de Mons, 20 place du Parc, 7000 Mons, 
Belgium}\\

{\em ${}^{b}$ Lebedev Institute of Physics,\\
Leninsky ave. 53, 119991 Moscow, Russia}\\

{\em ${}^{c}$ Institute for Theoretical and Mathematical Physics,\\
Lomonosov Moscow State University, 119991 Moscow, Russia}\\

\begin{abstract}
    Conformal Higher Spin Gravity is a higher spin extension of Weyl gravity and is a 
    family of local higher spin theories,
    which was put forward by Segal and Tseytlin. We propose a manifestly
    covariant and coordinate-independent action for these theories. The result is based on an interplay
    between higher spin symmetries and deformation quantization: a locally equivalent but manifestly background-independent reformulation, known as the parent system, of the off-shell multiplet of conformal higher spin fields (Fradkin--Tseytlin fields) can be interpreted in terms of Fedosov deformation quantization of the underlying cotangent bundle. 
    This brings into the game the invariant quantum trace, induced by the Feigin--Felder--Shoikhet cocycle of Weyl algebra, which extends Segal's action into a gauge invariant and globally well-defined action functional on the space of configurations of the parent system. 
    The same action can be understood within the worldline approach
    as a correlation function in the topological quantum mechanics
    on the circle.
\end{abstract}

\end{center}
\newpage
\tableofcontents
\newpage
\pagenumbering{arabic}
\setcounter{page}{1}

%%%%%%%%%%%%%%%%%%%%%%%%%%%%%%%%%%%%%%%%%%%%%%%%%%%%%%%%%%%%%
\section{Introduction}
%%%%%%%%%%%%%%%%%%%%%%%%%%%%%%%%%%%%%%%%%%%%%%%%%%%%%%%%%%%%%

There are very few examples of (covariant) actions for higher spin
(HS) gravities at present. (1) In three dimensions, there is a class
of topological higher spin theories
\cite{Blencowe:1988gj,Bergshoeff:1989ns,Pope:1989vj,Fradkin:1989xt,Campoleoni:2010zq,Henneaux:2010xg,Grigoriev:2019xmp,Grigoriev:2020lzu},
% This class
which encompasses massless, partially-massless and conformal
(higher spin) fields. The actions for these theories are simply the Chern--Simons
action for various Lie algebras that can be thought of as the higher spin
extensions of Poincar\'e, (anti-)de Sitter or conformal algebras.
(2) In the light-cone gauge, Chiral higher spin gravity in flat space
admits a very simple action
\cite{Metsaev:1991mt,Metsaev:1991nb,Ponomarev:2016lrm,Skvortsov:2018jea,Skvortsov:2020wtf}.
The theory has two contractions \cite{Ponomarev:2017nrr}, which
can be understood as higher spin extensions of self-dual Yang--Mills
and of self-dual gravity theories, %. The contractions
and have simple covariant actions \cite{Krasnov:2021nsq}. 
Some recent progress has been made towards uplifting Chiral theory
to twistor space \cite{Tran:2021ukl,Tran:2022tft,Herfray:2022prf,Adamo:2022lah}, where the Chern--Simons action \cite{Tran:2022tft}
captures correctly the cubic interactions both in flat and (A)dS spaces.
(3) IKKT model for a higher spin algebra \cite{Sperling:2017dts,Fredenhagen:2021bnw,Steinacker:2022jjv}
is an example of a non-commutative field theory with higher spin
fields in the classical limit. (4) The last but not the least is
the class of conformal higher spin (CHS) gravities
\cite{Tseytlin:2002gz, Segal:2002gd, Bekaert:2010ky}, which is
the subject of this paper. 

Conformal higher spin gravities (CHS gravity) are higher spin extensions
of conformal gravity. In four dimensions, it is an extension
of Weyl gravity. More generally, conformal gravities are theories
of a metric $g_{\mu\nu}(x)$ that are invariant under diffeomorphisms
and local Weyl rescalings
\begin{align}
    g_{\mu\nu}(x)\to g'_{\mu\nu}(x) = \Omega^2(x)\,g_{\mu\nu}(x)\,.
    \label{eq:Weyl_rescaling}
\end{align}
They exist in even dimensions $n\geq4$ and the cases of $n=2$
and $n=3$ are somewhat special, the latter one admitting
a Chern--Simons formulation \cite{vanNieuwenhuizen:1985cx,Horne:1988jf},
which was discussed above. In $n=4$ dimensions there is a unique action,
the Weyl action, while for $n>4$ there is an ambiguity growing
with dimension in what conformal gravity is, which is related
to the growing number of conformal invariants (see e.g.
\cite{Bonora:1985cq, Deser:1993yx, Boulanger:2004zf} for explicit
expressions in $6$ and $8$ dimensions, and also \cite{Boulanger:2018rxo}
for the recent progress concerning the general classification). 

It is remarkable that diffeomorphisms and Weyl rescalings
can be extended to an infinite-dimensional algebra of symmetries
that acts on an infinite multiplet of conformal higher spin fields
\cite{Fradkin:1989yd,Segal:2002gd}, which are also known as
Fradkin--Tseytlin fields \cite{Fradkin:1985am}. The multiplet
contains fields with arbitrarily high spin. Conformal higher spin
gravities are theories for this multiplet and there seems to be only
one such theory for every multiplet in even dimensions $n=2p$, $p\geq2$.
Two seemingly different constructions were proposed: Tseytlin's that is based on the effective action approach
\cite{Tseytlin:2002gz}(see also \cite{Bekaert:2010ky});
Segal's that is tight to a worldline model and deformation
quantization \cite{Segal:2002gd}. Both of these approaches
are eventually closely related to each other and prove the existence
of CHS gravities. 

One aspect of CHS gravity and, more generally, of ``higher spin geometry''
we would like to improve on is to propose a manifestly covariant
and both coordinate- and background-independent construction
for CHS gravities, including the action principle. Having such
a formulation is important for any extension of gravity and should
facilitate the study of these theories in the future, e.g. propagation
of CHS fields on gravitational backgrounds
\cite{Nutma:2014pua, Grigoriev:2016bzl, Beccaria:2017nco, Kuzenko:2019ill, Kuzenko:2019eni, Kuzenko:2020jie, Kuzenko:2022hdv}.

The problem of covariantization of CHS gravity leads to,
roughly speaking, two problems:
(i) what is the proper higher spin analog of the covariant derivative $\nabla$?
(ii) what is the proper higher spin analog of $\sqrt{g}$ to be able
to integrate? None of these objects makes sense a priori when
higher spin gauge fields are present. Indeed, higher spin gauge
transformations, of which the spin-two symmetries (diffeomorphisms
and Weyl rescalings in our case) form a small subset, mix spins
and derivatives and, hence, neither of $\nabla$ and $\sqrt{g}$ transform
in a meaningful way. In a broader sense, the real question is
``what is higher spin geometry?''.

A key step to the solution was given already in \cite{Segal:2002gd}.
Spin-two (low-spin) symmetries such as diffeomorphisms,
Weyl and Yang--Mills symmetries, can be represented by operators
of the first order at most (in spacetime derivatives). Higher spin
transformations bring in higher derivatives. Therefore, the natural
language is that of differential operators. The latter can also
be represented locally as the Moyal--Weyl star-product algebra.
The action of CHS gravity is then a specific invariant functional
on the Moyal--Weyl algebra, which takes advantage of the invariant trace $\Tr(\bullet)$.

Moyal--Weyl star-product is defined in terms of Darboux coordinates and does not support diffeomorphism symmetry.\footnote{In the sense that its very definition on a symplectic
manifold relies on a choice of coordinates --- called Darboux coordinates ---
wherein the symplectic form has constant components, and hence the star-product is defined only locally.} It known how to fix this problem within
the framework of deformation quantization --- one needs
to resort to the Fedosov approach \cite{Fedosov:1994zz},
which is  based on picking a background connection.
From the gauge field theory perspective, resorting to
Fedosov quantization amounts to the so-called parent
reformulation of the system. In the case of the Segal system,
the corresponding parent reformulation is known~\cite{Grigoriev:2006tt}
(see also \cite{Barnich:2004cr,Barnich:2010sw}; strictly speaking
we employ a certain partially gauge-fixed version of this formulation)
and its equations of motion are those of the Fedosov-like
quantization of the corresponding constrained system defined
on the cotangent bundle of the space-time manifold. 
The crucial point of this parent system is that it is
background-independent because the Fedosov-like connection
becomes a genuine gauge field and hence no background fields
are needed in the construction. As a result, we get an off-shell
gauge theory that contains the CHS multiplet together with
all the necessary auxiliary and pure gauge fields that encode
derivatives thereof in a higher spin covariant way. 

The last ingredient is Feigin--Felder--Shoikhet cocycle \cite{Feigin2005}
that allows one to define the invariant trace $\Tr(\bullet)$
on the Fedosov-quantized symplectic manifold.
This cocycle is a by-product of
Shoikhet--Tsygan--Kontsevich formality \cite{Tsygan1999, Shoikhet:2000gw}.
As a result, Segal's action, i.e. CHS gravity action in Darboux coordinates,
is a particular gauge of our action and one can choose other gauges that make various aspects manifest, e.g. one can pursue metric-like or frame-like approaches. 

It was shown in \cite{Grady:2015ica} that the Feigin--Felder--Shoikhet
cocycle can be represented as a specific correlator in a topological
quantum mechanics of a free particle on the circle. At the same time, Segal's action also admits a worldline
formulation \cite{Segal:2001di, Segal:2002gd, Bonezzi:2017mwr}. Therefore,
it should not be surprising that our covariant CHS gravity action can be
represented as a correlation function in a worldline model. 

The outline of the paper is as follows.
In Section \bref{sec:CHSreview}, we introduce conformal higher spin
fields, Fradkin--Tseytlin fields, and review the two approaches
to CHS gravities. In Section \bref{sec:covariant-action}, we introduce the parent extension of the Segal off-shell system and present
our main result --- a covariant action for CHS gravities. In Section~\bref{sec:gauges}, we discuss possible gauge conditions which allow us to reformulate the action in terms of the independent and unconstrained fields. We also elaborate on the explicit relation between frame-like and metric-like formulations of the system, which are derived from the parent formulation through suitable gauge conditions.
We end up with some conclusions and a discussion. 

%%%%%%%%%%%%%%%%%%%%%%%%%%%%%%%%%%%%%%%%%%%%%%%%%%%%%%%%%%%%%
\section{Conformal higher spin gravity}\label{sec:CHSreview}
%%%%%%%%%%%%%%%%%%%%%%%%%%%%%%%%%%%%%%%%%%%%%%%%%%%%%%%%%%%%%
We begin by outlining the problem of conformal higher spin gravity
and a closely related issue of a higher spin extension of conformal
geometry. Next, we review two approaches, Tseytlin's and Segal's,
to CHS gravity. At the end, we discuss the relation between
the two. While each of the approaches gives a `proof of concept’
for the existence of theory, we will see that obtaining a covariant
and globally well-defined form of it may not be straightforward,
and will propose a solution to this problem in section 
\bref{sec:covariant-action}.

%%%%%%%%%%%%%%%%%%%%%%%%%%%%%%%%%%%%%%%%%%%%%%%%%%%%%%%%%%%%%
\subsection{Conformal higher spin fields}
%%%%%%%%%%%%%%%%%%%%%%%%%%%%%%%%%%%%%%%%%%%%%%%%%%%%%%%%%%%%%
As previously stated, conformal higher spin gravity is an extension
of conformal, or Weyl, gravity, whose spectrum contains fields
of all integer spins. One way to describe Weyl gravity is in terms
of an equivalence class of conformal metrics, i.e. a rank-$2$ symmetric
tensor $g^{\mu\nu}$, subject to the gauge transformations
\begin{align}
    \label{spin2conf}
    \delta_{\xi,\sigma} g^{\mu\nu}= \mathcal{L}_\xi g^{\mu\nu}
    + 2\sigma g^{\mu\nu}\,,
\end{align}
where $\xi\equiv \xi^\mu(x)\,\partial_\mu$ is a vector field that generates infinitesimal diffeomorphism and $\sigma(x)$ is an arbitrary function parameterizing infinitesimal Weyl rescalings \eqref{eq:Weyl_rescaling}.

Conformal higher spin fields in the metric-like approach
\cite{Fradkin:1985am} are a natural generalization of the above
gauge symmetry to totally-symmetric tensors of arbitrary ranks $s>2$. More precisely,
conformal higher spin fields can be described by symmetric tensors
$\Phi^{\mu_1 \dots \mu_s}$ with $s \in \N$. The theory turns out
to be uniquely fixed by its (higher spin) non-abelian gauge
symmetries whose exact form can be read off from a simple matter
coupling and is reviewed in Section \bref{sec:Segal}. To envisage
the theory, it can be helpful to look for natural gauge symmetries
for $\Phi^{\mu_1 \dots \mu_s}$. To this effect, it is important
to isolate the spin-two field $g^{\mu\nu}$ and treat it
as a background. A higher spin extension of \eqref{spin2conf}
can be looked for, starting from
\begin{align}\label{proposal}
    \delta_{\xi,\sigma} \Phi^{\mu_1 \dots \mu_s}
    = \mathcal{L}_\xi \Phi^{\mu_1 \dots \mu_s}
    + w_s\,\sigma\,\Phi^{\mu_1 \dots \mu_s}
    + \nabla^{(\mu_1}\xi^{\mu_2 \dots \mu_s)}
    + g^{(\mu_1\mu_2}\sigma^{\mu_3 \dots \mu_s)} + \dots\,.
\end{align}
Here, the first term, the Lie derivative, declares
$\Phi^{\mu_1 \dots \mu_s}$ to transform as a tensor under diffeomorphisms
and the second term assigns a certain Weyl weight $w_s$ to it;
the third term is a generalization of diffeomorphisms to higher spins
(at this point the spin-two background is important and we use 
$\mathcal{L}_\xi g^{\mu\nu}=\nabla^\mu \xi^\nu+\nabla^\nu \xi^\mu$
to propose a higher spin extension); the fourth term represents
a higher spin Weyl transformation, but again over the spin-two background;
the dots $\ldots$ at the end are meant as a reminder that this expression needs
to be completed with terms that are nonlinear in higher spin fields
themselves. 

Upon linearization around flat space, the diffeomorphisms
get reduced to the Poincar\'e symmetries and the higher spin gauge
transformations become\footnote{This type of conformal gauge fields was first introduced in $4d$ \cite{Howe:1981qj} as sources to traceless conserved tensors.}
\begin{align}\label{usualFTtraceful}
    \delta_{\xi,\sigma} \Phi_{a_1 \dots a_s}
    & = \pl_{(a_1}\xi_{a_2 \dots a_s)}
    +\eta_{(a_1a_2} \sigma_{a_3 \dots a_s)}\,,
\end{align}
where $\eta_{ab}$ is the flat metric.\footnote{We use Greek indices
$\mu, \nu, \dots$ to denote (co)tangent indices on a (generally)
curved manifold and Latin indices $a,b,\dots$ refer to flat space
or fiber indices.} A free action for such fields
is simply given by\footnote{Generalizations to supersymmetric cases \cite{Fradkin:1985am, Fradkin:1989md, Fradkin:1990ps, Kuzenko:2017ujh, Kuzenko:2019ill, Kuzenko:2020jie, Kuzenko:2021pqm,Kuzenko:2022hdv,Kuzenko:2022qeq} and mixed-symmetry fields \cite{Vasiliev:2009ck} are also possible, but will be studied elsewhere. In the present paper, we deal with the most basic example of the bosonic CHS gravity with only totally symmetric fields in the spectrum.}
\begin{equation}\label{chsaction1}
    S_s[\Phi] = \int d^nx\,\Phi^{a_1 \dots a_s}\,
    P_{a_1 \dots a_s}^{b_1 \dots b_s}(\partial)\,
    \Box^{s+\frac{n-4}2}\,\Phi_{b_1 \dots b_s}\,,
\end{equation}
where $P^{b_1 \dots b_s}_{a_1 \dots a_s}(\partial)$ is a traceless
and transverse projector, thereby ensuring the invariance 
of this action under the linear gauge symmetries \eqref{usualFTtraceful}.
For example, for $s=1$ we have $P_a^b= \delta^b_a- \pl_a\pl^b/\square$.
The right power of $\square$ in \eqref{chsaction1} is to ensure locality
of the action. 
Upon integrating by part, this action can be brought to the form
\begin{equation}\label{chsaction2}
    S_s[\Phi] = (-1)^s\,\int d^nx\,C^{a_1 \dots a_s, b_1 \dots b_s}\,
    \Box^{\frac{n-4}2}\,C_{a_1 \dots a_s, b_1 \dots b_s}\,,
\end{equation}
where $C_{a_1 \dots a_s, b_1 \dots b_s}$ is the (linearised)
spin-$s$ Weyl tensor, which is defined as the projection of 
$\partial_{a_1} \dots \partial_{a_s} \Phi_{b_1 \dots b_s}$
onto the traceless rectangular Young diagram
$\parbox{25pt}{\Yboxdim{5pt} \tiny \gyoung(_5,_5)}$ of length $s$.
For spin-$2$, this reproduces the usual linearised Weyl tensor
and hence the linearised action of Weyl gravity.

One can proceed starting from the free actions \eqref{chsaction1} or
\eqref{chsaction2} and look for cubic and higher interaction vertices
while deforming the free gauge symmetry \eqref{usualFTtraceful}
accordingly, which is often called Noether procedure (or gauging in supergravity).
While this approach should, in principle, allow one to address the problem
of constructing CHS gravity in a systematic way, it is notoriously
difficult in practice, see \cite{Fradkin:1989md, Fradkin:1990ps,Metsaev:2016rpa,Kuzenko:2019eni,Kuzenko:2020jie,Kuzenko:2022hdv} for some results in this direction. 

In general, the introduction of interactions is strongly constrained by
the requirement of Weyl invariance, and in particular,
a simple dimensional argument shows that all possible vertices
compatible with conformal symmetry and involving a finite number of fields will contain a finite number
of derivatives, and hence the resulting theory will be local. Indeed,
within the perturbative expansion over flat space the conformal dimension
of $\Phi_{a_1 \dots a_s}$ is $(2-s)$ and an interaction vertex for
$k$ fields with spins $s_i$ schematically reads
\begin{align}
    S_k \sim \int d^nx\, (\partial)^p\Phi^{(s_1)} \ldots \Phi^{(s_k)}\,.
\end{align}
Therefore, it has to have a fixed number of derivatives $p=n+\sum_i (s_i-2)$.
Unfortunately, the price to pay for locality is that CHS gravity
is not unitary (as could be expected from the fact
that its kinetic terms are of higher derivative type),
which is true already for its spin-two subsector, Weyl gravity.
Nevertheless, it is an interesting example of higher spin gravity
that has applications within AdS/CFT duality
\cite{Liu:1998bu, Metsaev:2009ym, Bekaert:2012vt, Giombi:2013yva, Bekaert:2013zya, Chekmenev:2015kzf, Beccaria:2016tqy, Bekaert:2017bpy} 
and conformal geometry, 
since it produces many conformally invariant operators.
The S-matrix has good chances to be $1$ due to the higher spin
symmetry \cite{Joung:2015eny, Beccaria:2016syk} and, hence,
the theory may turn out to be unitary/trivial, which is
a sign of integrability.

Leaving aside the question of classifying possible CHS gravities,
there are two concise recipes to construct specific examples
of such theories that we are going to review now. These are Tseytlin's
approach \cite{Tseytlin:2002gz} that is based on the idea
of effective action and Segal's approach \cite{Segal:2002gd}
that draws inspiration from studying the quantized particle model
coupled to background higher spin fields. 

%%%%%%%%%%%%%%%%%%%%%%%%%%%%%%%%%%%%%%%%%%%%%%%%%%%%%%%%%%%%%
\subsection{Tseytlin's approach: induced action}
%%%%%%%%%%%%%%%%%%%%%%%%%%%%%%%%%%%%%%%%%%%%%%%%%%%%%%%%%%%%%
The approach due to Tseytlin \cite{Tseytlin:2002gz}
(see also~\cite{Bekaert:2010ky} for a more elaborated
exposition) rests on the idea of induced actions, already
discussed by Sakharov \cite{Sakharov:1967pk} for gravity.
It consists in considering a (complex) conformal scalar field
$\phi$ with (higher spin) currents $J_{a_1 \dots a_s}$ coupled
to a background of conformal higher spin fields, i.e.
\begin{equation}\label{coupledscalar}
    S_h[\phi] = \int d^nx\,\big(\phi^*\,\Box\,\phi
    + \sum_{s=0}^\infty\,J_{a_1 \dots a_s}\,h^{a_1 \dots a_s}\big)
    = \int d^nx\,\phi^*\,\big(\Box + \hat H\big)\,\phi\,.
\end{equation}
Here, $J_{a_1 \dots a_s}$ are bilinear in the scalar field $\phi$
and read schematically
\begin{align}
    J_{a_1 \dots a_s} & = \phi^* \pl_{a_1} \dots 
    \pl_{a_s} \phi + \ldots\,,
\end{align}
where $\ldots$ complete it to a traceless (can be achieved
off-shell) and conserved (on-shell) tensor. The fields
$h^{a_1 \dots a_s}$ can be treated as sources for these currents
or, and this is an interpretation we need, as background fields.
In the last step, we took the liberty to integrate by parts
so that no derivatives act on $\phi^*$ and all the derivatives
involved in the definition of the currents are distributed
over $\phi$ and $h^{a_1 \dots a_s}$. As a result of
this resummation, $\hat H$ is the (higher order) differential
operator acting on $\phi$ 
\begin{align}
    \hat H &= \sum_k \hat H^{a_1 \dots a_k}(h)\,
    \pl_{a_1} \dots \pl_{a_s}\,,
\end{align}
whose coefficients depend on the original sources
$h^{a_1 \dots a_s}$. Let us note for the future that
$\hat H=\hat H (x,\pl)$ is, basically, a generic formally
Hermitian\footnote{By `formally' Hermitian here we simply
mean that $F^\dag=F$, where $f(x)^\dag=f^*(x)$, $\pl_a^\dag=-\pl_a$ and $\dag$
is an anti-involution, $(FG)^\dag= G^\dag F^\dag$.}
differential operator. Since the higher spin currents are conserved
and traceless the background fields $h^{a_1 \dots a_s}$ enjoy
the same symmetry as \eqref{usualFTtraceful}. It is important
to stress that the simple Noether coupling \eqref{coupledscalar}
is not off-shell gauge invariant under \eqref{usualFTtraceful}
and, as usual, requires higher order corrections to the gauge
symmetry and to the Noether coupling (so-called Seagull terms). In the simplest
low spin cases, the complete coupling for the spin-one background
field $A_\mu$ originates from $D_\mu \phi^* D^\mu \phi$
with $D_\mu=\pl_\mu +i A_\mu$ and for the spin-two background,
i.e. the conformal metric itself, the stress-tensor coupling
$T_{ab} h^{ab}$ has to be appended with infinitely many terms
that can be resummed into the action of the conformally coupled
scalar field
\begin{align}
    S[\phi] & = \int d^nx\,\sqrt{g}\,
    (g^{\mu\nu}\pl_\mu \phi^* \pl_\nu\phi
    + \tfrac{(n-2)}{4(n-1)}\phi^* R \phi)\,.
\end{align}

Suppose we are given an action $S_h[\phi]$ for
the conformal scalar field $\phi$ coupled to an arbitrary
conformal higher spin background (collectively denoted by $h$).
One can then compute the effective action
\begin{equation}
    e^{-W[h]} = \int\cD\phi^*  \cD\phi\,e^{-S_h[\phi]}\,,
\end{equation}
using the heat kernel method
\begin{equation}
    W[h] = -\int_\epsilon^\infty \tfrac{dt}t\,\Tr\,e^{-t\hat F}\,,
    \qquad
    \hat F := \Box+\hat H\,,
\end{equation}
where $\epsilon$ is a cut-off. Using the heat kernel expansion,
the effective action can be written as
\begin{equation}
    W[h] = (\text{poles in $\epsilon$})
    + a_{n/2}[\hat F]\,\log\epsilon
    + (\text{series in $\epsilon$})\,,
\end{equation}
where the Seeley--DeWitt coefficient $a_{n/2}[\hat F]$
of the operator $\hat F=\hat F[h]$ appears only in even
dimension $n$. The pole part represents the usual UV divergences. 
The coefficient $a_{n/2}$ is a local functional of the conformal
higher spin fields $h$ 
and is Weyl invariant. As such it has all the desired properties
that one would ask of an action for CHS gravity.
In other words, the CHS gravity action in Tseytlin's
approach is defined as the logarithmically divergent piece
of the effective action of a scalar field conformally coupled
to a background of CHS fields. It is also well-known
that for the low spin background fields $A_\mu$ and $g^{\mu\nu}$
in, say $n=4$, we have \cite[Sec. 3]{Beccaria:2017nco}
\begin{align}
    a_{(n=4)/2} & = \tfrac{1}{(4\pi)^2}\int d^4x\, \sqrt{g}
    \left(-\tfrac1{12} F_{\mu\nu}F^{\mu\nu}
    + \tfrac{1}{120} C_{\mu\nu,\lambda\rho}C^{\mu\nu,\lambda\rho}\right)\,,
\end{align}
where $C_{\mu\nu,\lambda\rho}$ is the Weyl tensor, 
and the topological Euler term was dropped. Therefore,
it should not be too surprising that $a_{n/2}$ receives well-defined,
local higher spin corrections as long as we switch on the higher spin
background fields. One subtlety can be that $\hat F$ is a higher derivative
operator (possibly of infinite order) and the Heat kernel techniques
for general higher derivative operators are yet to be developed
(see e.g. \cite{Barvinsky:2019spa, Barvinsky:2021ijq} for recent works
on heat kernels for higher-order operators). The induced action approach should be generalizable to mixed-symmetry and supersymmetric cases. In particular, it has recently been advanced to $\mathcal{N}=1$ supersymmetric conformal higher spin case \cite{Kuzenko:2022hdv,Kuzenko:2022qeq} and the quadratic part of the action has been obtained via the induced action technique \cite{Kuzenko:2022qeq}.

%%%%%%%%%%%%%%%%%%%%%%%%%%%%%%%%%%%%%%%%%%%%%%%%%%%%%%%%%%%%%
\subsection{Segal's approach: particle in a higher spin background}
\label{sec:Segal}
%%%%%%%%%%%%%%%%%%%%%%%%%%%%%%%%%%%%%%%%%%%%%%%%%%%%%%%%%%%%%
Suppose we are given the quantized phase-space, i.e. an associative
algebra of functions on the phase space $\R^{2n}$ with linear coordinates
$x^\mu$, $p_\nu$, subject to $\pb{x^\mu}{p_\nu}=\delta^\mu_\nu$. 
Here and in what follows we employ the language of Weyl symbols,
i.e. identify the operator algebra of polynomial functions $p$ with coefficients in smooth functions in $x$ 
tensored with $\R[[\hbar]]$ (formal power series in $\hbar$)
and with the product being Moyal--Weyl star-product, 
\begin{equation}\label{mw}
    (f \star g)(x,p) = f(x,p)\,\exp\Big[\tfrac\hbar2\,
    (\tfrac{\overset{\leftarrow}{\partial}}{\partial x^\mu}\,
    \tfrac{\overset{\rightarrow}{\partial}}{\partial p_\mu}
    -\tfrac{\overset{\leftarrow}{\partial}}{\partial p_\mu}\,
    \tfrac{\overset{\rightarrow}{\partial}}{\partial x^\mu})\Big]\,g(x,p)\,,
\end{equation}
so that in particular, the commutation relations between
the coordinates $x^\mu$ and $p_\nu$ read
\begin{equation}
\label{xp}
    [x^\mu,p_\nu]_\star = \hbar\,\delta^\mu_\nu\,.
\end{equation}
This algebra admits an anti-involution defined by
\begin{equation}
    x^\dagger = x\,,
    \qquad
    p^\dagger = p\,,
    \qquad
    \hbar^\dagger = -\hbar\,, \qquad
    (a \star b)^\dagger = b^\dagger \star a^\dagger\,.
\end{equation}
We also employ the standard representation of the algebra
on functions in $x^\mu$ tensored with $\R[[\hbar]]$. More precisely,
any $f(x,p)$ is sent to a differential operator $\hat f(x,\dl{x})$
that is obtained by employing the symmetric ordering. This gives
a quantization map. Its inverse, sending 
operators to phase space functions is usually referred to as the symbol map.
Under the quantization map, the involution $\dagger$ is sent to
the formal adjoint with respect to the standard inner product
\begin{align}
    \inner{\phi}{\psi} = \int d^nx \,\phi^*\,\psi\,.
\end{align}
The above construction is invariant under general linear transformation
of $x$ complemented by the conjugate transformations of $p$.
In particular, `wave-functions' $\phi$, $\psi$ are assumed to be
semi-densities for the inner product to be invariant under
such transformations. Note that more general diffeomorphisms
cannot be naturally implemented in this framework since, for instance,
they do not preserve Moyal--Weyl star-product \eqref{mw}.

Let $F(x,p)$ be a generic element of the algebra which we view
as a first class constraint (we do not explicitly indicate the $\hbar$
dependency, having in mind that we work over $\R[[\hbar]]$). 
Assuming $F$ Hermitian and polynomial in $p_a$, one can easily
write the free action of the associated scalar field as
\begin{equation}
\label{scalarcoupled}
    S[\phi,F] = \inner{\phi}{\hat F \phi}\,.
\end{equation}
In this action, we view $F$ as a generating function for background
fields, and $\phi=\phi(x)$ as a complex scalar field. Indeed,
the relation to the previous section is that in the action
\eqref{scalarcoupled} we can absorb $\square$ into
$\hat{F}=\square+\hat{H}$, i.e. $\eta^{\mu\nu}\,p_\mu p_\nu$
is a (spin-two) background value for $F$ and the perturbations
correspond to turning on a conformal higher spin background.
In particular, it is more natural to associate the Taylor
coefficients of $F$ when expanded in $p$
\begin{align}
    F &= \sum_s F^{\mu_1 \dots \mu_s}(x)\,
    p_{\mu_1} \dots p_{\mu_s}\,,
\end{align}
with the background fields (recall that they can be obtained
via a certain rearrangement of $h^{a_1 \dots a_s}$).

The advantage of the action \eqref{scalarcoupled} is that
it is very easy to determine the full gauge symmetry that leaves it
invariant. Indeed, it is obvious that it is invariant under
the following infinitesimal gauge symmetries,
\begin{subequations}
\label{Segalgauge}    
\begin{align}\label{Segaloffshell}
    \delta_{\xi,w} F & = \frac{1}{\hbar}\,[F,\xi]_\star
    + \{F,w\}_\star\,, \\
     \delta_{\xi,w} \phi &= -\left(\tfrac{1}{\hbar}\,\hat{\xi}
    +\hat w\right)\phi\,,
\end{align}
\end{subequations}
where the gauge parameters $\xi$ and $w$ are Weyl symbols of the Hermitian operators: $\hat\xi^\dagger=\hat\xi$, $\hat w^\dagger =\hat w$ and where ``hat'' denotes the quantization map. These symmetries have a natural interpretation in terms of
the Hamiltonian constrained system describing the underlying particle model, see Appendix~\bref{app:BRST} for details. An even more compact way to represent the same gauge symmetry is to define $u=\hbar^{-1} \xi + w$, which is neither Hermitian nor anti-Hermitian, and write
\begin{align}\label{usymmetries}
    \delta_u F & = u^\dag \star F + F \star u\,,
    && \delta_u \phi = -\hat{u}\, \phi\,.
\end{align}
In what follows we refer to the off-shell system
\eqref{Segaloffshell} with field $F$ subject to
the above gauge symmetries as the {\it off-shell Segal system}.
In other words, we drop the $\phi$-part. The off-shell Segal
system defines a completion of \eqref{proposal}
and also solves the problem raised in the previous section,
of how to couple matter fields to a higher spin background. 

Let us comment on the background independence of the system.
At first glance, the definition of the system does not involve
any background fields and hence the system can be considered as
a gravity-type theory (for the moment defined at the off-shell
level only). Nevertheless, a more careful analysis shows
that the star-product involved in the construction 
depends on a choice of Darboux coordinates. Although
this background dependence is not of a fundamental nature
and can, in principle, be excluded by working with differential
operators as such, it is not clear if this does not lead
to further complications. More geometrically,
what we are implicitly dealing with is a cotangent
bundle over the spacetime manifold $\manX$ and it is known that one, at least, needs
to fix an affine connection on $\manX$ to define a concrete
star-product in a coordinate invariant way. As we are going to see
in the next Section, this drawback can be cured by passing to
a locally equivalent reformulation of the system.

By analyzing the component form of the transformation
\eqref{Segaloffshell}, one finds that they generalize
gauge transformations of the conformal gravity and,
at the same time, give a nonlinear extension
of the linearized gauge symmetries of the Fradkin--Tseytlin
fields. Note however that to see that the linearization
of the above gauge symmetries indeed reproduces
the Fradkin--Tseytlin gauge transformations,
one is to employ an intricate field redefinition,
introduced in~\cite{Segal:2002gd} (see also 
\cite{Grigoriev:2006tt,Bekaert:2010ky}). 
Let us give a few sketchy arguments supporting these statements. 
First of all, we note that the above gauge transformations
contain a subalgebra of diffeomorphisms and Weyl rescalings 
which act on the spin-2 component of $F$ in a standard way,
but at the same time affects other components of $F$.
More precisely, consider transformations with the parameters
\begin{align}\label{conformalpart}
    \xi&= \xi^\mu(x) p_\mu\,,  & w&=\sigma(x)\,.
\end{align}
The commutator of two such transformations, with
$u_i=\tfrac1\hbar\,\xi_i^\mu(x) p_\mu +\sigma_i(x)$
for $i=1,2$, reads as
\begin{align}\label{diffalgebra}
    [u_1, u_2]_\star &= \tfrac1\hbar\,
    (\xi_2^\mu \pl_\mu \xi_1^\nu- \xi_1^\mu \pl_\mu \xi_2^\nu) p_\nu
    + \xi_2^\mu \pl_\mu \sigma_1-\xi_1^\mu \pl_\mu \sigma_2\,,
\end{align}
and hence this subalgebra is isomorphic to the semidirect product
of the algebra of diffeomorphisms (represented by vector fields)
with the (abelian) algebra of Weyl rescalings. However,
this subalgebra is not represented in a canonical way
on the $p$-Taylor coefficients of $F$ as one can see
this already with the first even-spin components of the field $F$
\begin{align}
    F & = \bloody(x) + \tfrac12\,g^{\mu\nu}(x)\, p_\mu p_\nu
    + \dots\,,
\end{align}
where notation $g^{\mu\nu}$ will justify itself immediately.
Here we assume that the higher spin components vanish otherwise
there can be higher derivative contributions below. The subalgebra
\eqref{conformalpart} of the gauge transformations, which consists
of diffeomorphisms and Weyl rescalings, acts as 
\begin{align}\label{metrictr}
    \delta_{\xi,\sigma} g^{\mu \nu} & = -\pl_\lambda \xi^\mu g^{\lambda \nu}
    - \pl_\lambda \xi^\nu g^{\lambda\mu}
    + \xi^\lambda \pl_\lambda g^{\mu \nu}
    + 2\sigma g^{\mu \nu}\equiv \mathcal{L}_\xi g^{\mu\nu}
    + 2\sigma g^{\mu\nu}\,,\\ \label{entanglGauge}
    \delta_{\xi,\sigma} \bloody & = \xi^\mu \pl_\mu \bloody + 2\sigma \bloody
    - \frac12\hbar^2 \pl_\mu\pl_\nu \xi^\lambda \pl_\lambda g^{\mu \nu}
    - \frac12 \hbar^2g^{\mu \nu}\pl_\mu \pl_\nu \sigma\,.
\end{align}
We see that $g^{\mu\nu}$ transforms as a conformal metric should,
but $\bloody$ has additional non-covariant terms (the last two
terms, proportional to $\hbar^2$)
that get $\bloody$ entangled with $g^{\mu\nu}$. In principle,
there is a field-redefinition $\bloody\rightarrow \bloody+ Y(g)$
that allows one to eliminate such terms \cite{Segal:2002gd}.
However, it is, generally, difficult to find one and it becomes more and
more complicated to disentangle transformations as we proceed
to higher spins. While some formulas, e.g. \eqref{diffalgebra}
and \eqref{metrictr}, indicate that the diffeomorphism algebra
is a part of the off-shell Segal system, others such as \eqref{mw}
and \eqref{entanglGauge} make it very difficult to see
how it is realized on fields.\footnote{Retrospectively, 
this is consistent with the fact that we are dealing with
symbols of differential operator, for which it is known
that only the principal part (coefficients of its term
of highest order in derivatives) transforms as a totally
symmetric tensor (see e.g. \cite{Bekaert:2021sfc}
for more details on this approach in relation to
higher spin gravity).} This is one of the problems
that the present paper is to solve: how to covariantize
the Segal approach.\footnote{For the moment we only talk
about the off-shell Segal system, whose covariantization
is known \cite{Grigoriev:2006tt,Grigoriev:2016bzl}.}

Genuine higher spin gauge transformations \eqref{usualFTtraceful}
can be seen once we expand \eqref{Segalgauge} over the background
$F^{(0)}=\tfrac12\,p^2 \equiv \tfrac12\,\eta^{\mu\nu}\,p_\mu p_\nu$. 
Indeed, for $F = \tfrac12\,p^2 + f$,
the linearized gauge transformations read
\begin{align}
    \delta_{\xi,w} f & = \tfrac{1}{2\hbar}\,[p^2,\xi]_\star
    + \tfrac12\,\{p^2,w\}_\star\,,
\end{align}
and for generic $\xi$ and $\sigma$ we find
\begin{align}
    \delta_{\xi,w} f(x;p) & = -p^\mu \pl_\mu\xi(x;p)
    + p^2\,\sigma(x;p) + \ldots\,,
\end{align}
where $\ldots$ denotes corrections of order $\mathcal{O}(\hbar^2)$
which are also responsible for mixing the fields of different spins
and for the non-covariance. When Taylor expanded in $p$ the above formula
reproduces \eqref{usualFTtraceful} to the leading order in $\hbar$.
This way $\xi$ and $w$ represent {\it higher spin (HS) diffeomorphisms}
and {\it higher spin (HS) Weyl transformations}.

To summarize, the algebra of diffeomorphisms and Weyl rescalings
is obviously a part of the gauge symmetries of the off-shell Segal
system. In addition, the linearization over $\tfrac12\,p^2$
does reproduce the higher spin transformations \eqref{usualFTtraceful}.
Therefore, altogether we are in a possession of a certain completion
of \eqref{proposal}. Nevertheless, the diffeomorphisms and Weyl transformations are not
represented in the canonical way and the Taylor coefficients
of $F$ do not behave as tensors under diffeomorphisms. 

The off-shell Segal system can be expanded over any background,
the simplest one being $\tfrac12\,p^2$, which is also
the one making \eqref{scalarcoupled} into the action of
a conformal scalar field.

Given a fixed background solution, a natural question is
to consider an algebra of its symmetries.
By definition these are the leftovers of gauge symmetries
that leave it intact:
\begin{align}
    u^\dag \star p^2 + p^2 \star u & = 0\,.
\end{align}
We also note that the symmetries \eqref{usymmetries} are
reducible: $\delta_u F=0$ for $u=i v \star F$ with $v^\dag=v$.
Therefore, the symmetry algebra of the vacuum $\tfrac12\,p^2$
is exactly the algebra of `higher symmetries of Laplacian'
\cite{Nikitin1991,Eastwood:2002su,Segal:2002gd},
i.e. the quotient algebra of the symmetries modulo trivial ones.
At the same time, this algebra is the deformation quantization
\cite{Eastwood:2002su,Michel2014, Joung:2014qya} of the coadjoint
orbit corresponding to the free scalar field in flat space,
$\square \phi=0$, as a representation of the conformal algebra
$\mathfrak{so}(n,2)$. This is an infinite-dimensional associative algebra
$\hs(\square\phi)$ that can also be obtained as a quotient
of $\cU(\mathfrak{so}(n,2))$ by a certain two-sided ideal, known as
the Joseph ideal \cite{Joseph1976} (see also
\cite{Fernando:2015tiu, Gunaydin:2016bqx} and references therein
for further applications to the construction of various higher algebras). 

A simple generalization of this construction is to take the same
scalar matter $\phi(x)$, but with a different kinetic term,
$\square^k \phi$, which are called higher order singletons
\cite{Bekaert:2013zya}. Obviously, they correspond to
$(p^2)^k$ within the Segal construction. The corresponding
higher spin algebras $\hs(\square^k\phi)$ are not isomorphic
for different values of the integer $k$
\cite{Eastwood2008, Gover2009, Alkalaev:2014nsa, Michel2014, Joung:2015jza, Brust:2016zns},
the theories have different spectra of higher spin currents and,
hence, the background fields they couple to. Therefore,
one can have an access to a larger set of CHS gravities
by picking different vacua. It is a vital point of Segal's 
construction that even though most of the discussion does not
have to refer to any particular vacuum,
it is necessary to make a choice at some point since different 
vacua correspond to different spectra of CHS fields,
i.e. to essentially different theories.

As we discussed above, the action for the background fields
can be obtained as a log-divergent piece in the effective action
of a conformal scalar field coupled to the higher spin background.
Here, we recall that within the approach of Segal the action
is defined as
\begin{align}\label{segal-action}
    S[F] & = \int d^nx\, \cL_x(F)\,,
    && \cL_x(F) = \int d^np\, \cL_{x,p}(F)\,, 
\end{align}
where $\cL_x(F)$ is the Lagrangian that can be proved
to be a local function of the $p$-Taylor components of $F$
and $x$-derivatives thereof. The Lagrangian $\cL_x(F)$ results
from evaluating an integral over $p$ for an auxiliary Lagrangian
density $\cL_{x,p}(F)$ on the phase-space. In other words,
\begin{align}
    S[F] &= \tr\, \cL_{x,p}(F) =\int d^nx\, d^np\, \cL_{x,p}(F)\,, 
\end{align}
where $\tr{}$ denotes the usual trace in the Weyl quantization given by the phase-space integral of the respective symbol. The phase-space Lagrangian $\cL_{x,p}(F)$ is a specific star-product function satisfying
\begin{align}
    \cL^\prime_{x,p}(F) \star F & = 0 = F \star \cL^\prime_{x,p}(F)\,,
\end{align}
for any function $F$. It is easy to see that this property indeed guarantees
gauge invariance of~\eqref{segal-action}.  At the formal
level this $\cL_{x,p}$ function can be defined as 
the star-product Heaviside step-function,
\begin{align}
    \cL_{x,p}(F) & = \Theta_\star(F)\,.
\end{align}
While some care is needed to work with such an object,
see \cite{Segal:2002gd} and Appendix \bref{app:starHeaviside},
it can be shown to be well-defined.\footnote{For example,
there is also a sort of regularization involved into the definition
of $\Theta_\star(F)$ and one of the crucial steps is to pick
the vacuum $F^{(0)}=\tfrac12\,p^2$ and treat all the other fields
as small perturbations (large perturbations can reach
a non-equivalent vacuum). Also, evaluation of the $p$-integral
requires Euclidean signature, but once $\cL_x(F)$ is computed,
the corresponding action is gauge invariant for any choice
of the signature.}

Let us sketch the proof of the gauge invariance of the action.
There are two types of gauge transformations involved
in \eqref{Segalgauge}: (a) adjoint transformations via commutator
$[F,\xi]_\star$ and (higher spin) Weyl transformations
via anti-commutator $\{F,w\}_\star$. The invariance
under (a) is manifest due to the use of star-product
and of the invariant trace $\tr$, i.e. its cyclic property
$\tr[f,g]_\star=0$. The invariance under (b) is what fixes
$\cL_{x,p}$ to be the star-product Heaviside function,
since it requires
\begin{align}\notag
    \delta_w S & = \tr\,\delta_w\cL_{x,p}(F)
    = \tr\,\left(\cL_{x,p}'(F) \star \{F,w\}_\star\right)
    = 2\tr\,\left(\cL_{x,p}'(F) \star F \star w\right) = 0\,,
\end{align}
where we used the cyclicity of the trace. Similarly to
$\Theta'(x)x = \delta(x)x \equiv 0$, the last equality formally
implies $\cL_{x,p}(F)=\Theta_\star(F)+\text{const}$. The constant
does not give any contribution in the Segal case, but it does lead
to an `index' for the covariantized CHS gravity we discuss
in Section \bref{sec:covariant-action}.

%%%%%%%%%%%%%%%%%%%%%%%%%%%%%%%%%%%%%%%%%%%%%%%%%%%%%%%%%%%%%
\subsection{Tseytlin/Segal dictionary}
%%%%%%%%%%%%%%%%%%%%%%%%%%%%%%%%%%%%%%%%%%%%%%%%%%%%%%%%%%%%%
At first sight, it may seem that Tseytlin's and Segal's
constructions are completely unrelated. Indeed, computationally
this is true to an extent: in Tseytlin's approach one has
to extract the log-divergent pieces of one-loop Feynman
diagrams of matter fields with higher spin background fields
on external lines; in Segal's approach one is to expand
the star-product Heaviside step-function. Nevertheless,
one can give an argument \cite[Sec. 6.2]{Segal:2002gd}
(see also \cite[App. D]{Basile:2018eac}) for why the end result
has to be the same. 

Let us start from Tseytlin's definition,
\begin{equation}
    S[\hat F] = a_{\frac n2}[\hat F]\,,
\end{equation}
and re-write it in terms of the heat kernel for $\hat F$.
The latter admits a small $t$ expansion \cite{Bekaert:2010ky},
\begin{equation}
    \Tr\big(e^{-t\,\hat F}\big) = t^{-\frac n2}\,
    \sum_{k=0}^\infty\,t^k\,a_k[\hat F]\,,
\end{equation}
so that the CHS gravity action reads
\begin{equation}
    S[\hat F] = \tfrac1{2i\pi}\,
    \oint \tfrac{dt}t\,\Tr(e^{-t\,\hat F})\,,
\end{equation}
where the integral is over a closed contour including
the origin of the complex plane. Now recall that, given
two differential operators $\hat D_1$ and $\hat D_2$,
with symbols $D_1$ and $D_2$ respectively, the symbol
of their composition $\hat D_1 \circ \hat D_2$
is the star product of their symbols, $D_1 \star D_2$.
As a consequence, the symbol of $e^{-t\,\hat F}$
is given by the star-exponential of the symbol $F$
of $\hat F$, i.e.
\begin{equation}
    e_\star^{-t\,F} := \sum_{k=0}^\infty
    \tfrac{(-t)^k}{k!}\,F^{\star k}\,,
    \qquad
    F^{\star k} := \underbrace{F \star \dots \star F}_{k\,\text{times}}\,.
\end{equation}
On top of that, the trace of a differential operator
agrees with the trace of its symbol,
\begin{equation}
    \Tr(\hat D) = \tr(D)
\end{equation}
so that the CHS gravity action can be written as
\begin{equation}\label{eq:action_contour}
    S[\hat F] = \tfrac1{2i\pi}\,\oint\tfrac{dt}t\,
    \tr(e^{-t\,F}_\star)\,.
\end{equation}
Assuming that the trace of the star-exponential of $(-t\,F)$
is analytic in the $\Re(t)>0$ region of the complex $t$-plane,
we can deform the original closed contour around the origin
to a half-circle whose diameter consists of the line defined
by $\Re(t)=-\epsilon$ with $\epsilon\to0^+$
(see Fig. \ref{fig:contour}). Further assuming that
$\tr(e_\star^{-t\,F})$ goes to $0$ when
$\lvert t \rvert \to \infty$, the contribution of
the contour integral on the arc of the half-circle 
vanishes as the radius is sent to infinity, so that
the action \eqref{eq:action_contour} now reads
\begin{equation}
    S[\hat F] = \lim_{\epsilon\to0^+}\,
    \tfrac1{2i\pi}\,\int^{+\infty}_{-\infty}
    \tfrac{d\tau}{\tau-i\epsilon}\,\tr(e_\star^{i\tau\,F})\,.
\end{equation}
Finally, assuming that we can exchange the order of the trace
and the contour integral, the above action can be 
re-written as
\begin{equation}
    S[\hat F] = \tr\big(\Theta_\star(F)\big)\,,
\end{equation}
where
\begin{equation}
    \Theta_\star(a) := \lim_{\epsilon\to0^+}\,\tfrac1{2i\pi}\,
    \int_{-\infty}^{+\infty}\,\tfrac{d\tau}{\tau-i\epsilon}\,
    e_\star^{i\tau a}\,,
\end{equation}
is the Heavisde star-function, as introduced by Segal
\cite{Segal:2002gd}. Even though some of the steps followed
previously may seem a bit formal (in the sense that we had
to assume some analytical properties of $e_\star^{-t\,F}$
and its trace), Segal showed how to compute the quadratic
and cubic pieces of the above action with $F=\tfrac12\,p^2 + \cO(h)$,
where $h$ denote conformal higher spin fields (see Appendix
\bref{app:starHeaviside} for more details).

\begin{figure}[!ht]
    \center
    \begin{tikzpicture}
        \draw[help lines, color=gray!30, dashed] (-2.4,-2.4) grid (2.4,2.4);
        \draw[->,ultra thick] (-2.5,0)--(2.5,0);
        \draw[->,ultra thick] (0,-2.5)--(0,2.5);
        \node (second) at (-0.5,-0.5) {};
        \draw[thick, RoyalBlue] (0,0) circle (1cm);
          angle=15];
        \draw[->, thick, RoyalBlue] (0.70,-0.70) arc[radius=0.7cm,start angle=310,delta angle=12];
    \end{tikzpicture}
    \hspace{75pt}
    \begin{tikzpicture}
        \draw[help lines, color=gray!30, dashed] (-2.4,-2.4) grid (2.4,2.4);
        \draw[->,ultra thick] (-2.5,0)--(2.5,0);
        \draw[->,ultra thick] (0,-2.5)--(0,2.5);
        \draw[->, thick, RoyalBlue] (-0.5,2)--(-0.5,-2);
        \draw[->, thick, RoyalBlue] (-0.5,-2) arc (-90:90:2cm);
        \draw[->] (-0.5,0) -- (0.75,1.6);
        \node at (1.3,1.75) {$R \to \infty$};
    \end{tikzpicture}
    \caption{Deformation of the contour used to pick up
    the Seeley--DeWitt coefficient of the logarithmic
    divergence in the trace of the heat kernel.}
    \label{fig:contour}
\end{figure}
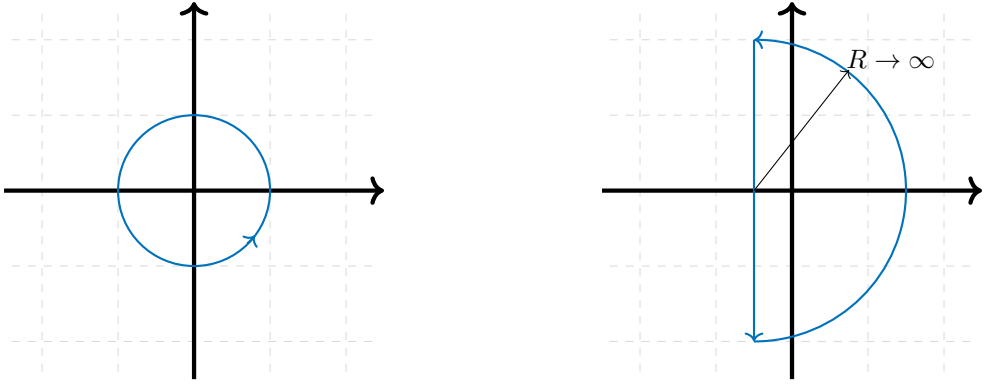

%%%%%%%%%%%%%%%%%%%%%%%%%%%%%%%%%%%%%%%%%%%%%%%%%%%%%%%%%%%%%
\section{Covariant action for conformal higher spin gravity}
\label{sec:covariant-action}
%%%%%%%%%%%%%%%%%%%%%%%%%%%%%%%%%%%%%%%%%%%%%%%%%%%%%%%%%%%%%
As it is clear from the review of the two approaches to CHS gravity
recalled in the previous section, there is no doubt that this class of theories
does exist and comprises well-defined local field theories. CHS gravity is,
first of all, a theory of gravity and, hence, it is important to have
the diffeomorphism/Weyl symmetry and higher spin extensions thereof
represented in a simple way. This is the goal of the present Section.

%%%%%%%%%%%%%%%%%%%%%%%%%%%%%%%%%%%%%%%%%%%%%%%%%%%%%%%%%%%%%
\subsection{Parent form of the off-shell Segal system}
\label{sec:parent_Segal}
%%%%%%%%%%%%%%%%%%%%%%%%%%%%%%%%%%%%%%%%%%%%%%%%%%%%%%%%%%%%%
The off-shell Segal system \eqref{Segalgauge} can be equivalently
reformulated in a 
coordinate-independent and globally
well-defined way with the help of a version of the Fedosov quantization
approach. From the field theory perspective, this amounts to 
what is known as the
parent reformulation~\cite{Barnich:2004cr,Grigoriev:2006tt,Barnich:2010sw}.
Its application to CHS gravity was already detailed
in \cite[App. A]{Grigoriev:2016bzl}.

The reformulation is constructed as follows.
We consider a formal version of the quantized phase space~\eqref{xp}.
This is obtained by taking Weyl algebra $\cA_{2n}$ of polynomials
in $p_a$ with coefficients in formal power series in $y^a$ with
the product being the Moyal--Weyl star-product, denoted $\ast$
so as to distinguish it from the previous section (as before
the algebra is considered as that over $\R[[\hbar]]$). 
Given an $n$-dimensional
space-time manifold $\manX$, consider a Weyl algebra bundle $W(\manX)$
over $\manX$, whose fiber at $x\in\manX$ is a Weyl algebra $\cA_{2n}$,
associated with $T_x\manX \oplus T^*_x\manX$. In other words , sections of this bundle are elements of 
$\Gamma\big(S(T\manX)\big) \otimes \Gamma\big(\hat S(T^*\manX)\big)$
where $S$ denote the symmetric algebra and $\hat S$
its completion.\footnote{{This distinction between the symmetric algebra of $T\manX$ and the completion of the symmetric algebra
of $T^*\manX$ reflects the fact that we are considering polynomials 
in $p$ (coordinates on the fibers of $T^*\manX$)
and formal power series in $y$ (coordinates on the fibers of $T^*\manX$).}} Consider then a $\cA_{2n}$-valued  connection $1$-form $A$ and a section $F$ of $W(\manX)$,
subject to the following
equations:\footnote{Note that a simple extension of this parent
system would be to require that the curvature of $A$ takes
values in the center of the Weyl algebra (i.e. is simply given by
a closed $2$-form) instead of vanishing, exactly like
in Fedosov quantization.}
\begin{align}
    \label{parentEOM}
    dA + \tfrac1{2\hbar}\,\qcommut{A}{A} & = 0\,,
    && dF + \tfrac1\hbar\,\qcommut{A}{F}=0\,,
\end{align}
and the following gauge symmetries
\begin{align}\label{gs}
    \delta_\xi A & = d\xi + \tfrac{1}{\hbar}\qcommut{A}{\xi}\,,
    & \delta_{\xi,w} F & =\tfrac{1}{\hbar}\qcommut{F}{\xi}+\scommut{F}{w}\,,
\end{align}
where $\xi$ and $w$ are sections of $W(\manX)$, i.e. locally they are zero-forms with values in $\cA_{2n}$, and $w$ satisfies
\begin{align}\label{covconstw}
    d w+\tfrac{1}{\hbar}\qcommut{A}{w}&=0\,.
\end{align}
Note that the gauge parameter $w$ is subject to a differential
constraint. In fact, it is equivalent to an algebraic one provided
that the term $dx^\mu e_\mu^a(x)p_a$ in $A$, which is linear in $p$ and $y$-independent, is determined by an invertible $e_\mu^a(x)$, as we assume below.
Indeed, this condition then allows one to uniquely express all
the components of $w$ in terms of $w|_{y=0}$, as is explained
in Section~\bref{sec:gauges}. We also assume that $g^{ab}(x)$ determining the term quadratic in $p$ and $y$-independent, $\half g^{ab}(x)p_a p_b$, in $F$
is invertible. 

The above system is obtained as a partial gauge fixing of the equivalent
(locally in the spacetime) representation of the off-shell Segal system
proposed in~\cite{Grigoriev:2006tt}. In particular, because the relevant
gauge transformation is algebraic and the gauge condition is strict,
the gauge fixing produces an equivalent system. Moreover, the above system
is also off-shell. The system \eqref{parentEOM}-\eqref{covconstw}
will be called {\it parent Segal system} and solves the problem
of covariantizing the off-shell Segal system \eqref{Segalgauge}. Note also that if one omits the gauge transformations with parameters $w$, equations \eqref{parentEOM} and \eqref{gs} are precisely the defining equations for the Fedosov-like connection and the lift of functions in the  version of the Fedosov quantization suitable for cotangent bundles. These equations were also discussed in~\cite{Vasiliev:2005zu,Grigoriev:2006tt} as an off-shell system for massless higher-spin fields.

It is easy to give an independent proof of the equivalence by employing
a special gauge where $A=dx^\mu \delta_\mu^a \, p_a$.\footnote{In particular, this implies picking a coordinate frame where $e^a_\mu=\delta^a_\mu$. In the Fedosov deformation quantization of general symplectic manifolds  the analogous gauge can be chosen locally and requires Darboux coordinates on the base manifold.} As such, this gauge is reachable locally,
see Section \bref{sec:gauges} and Appendix~\bref{app:Fedosov}
for details. We will also call it Segal's gauge, for it is easy
to make contact with the Segal approach of Section \bref{sec:Segal}.
Indeed, in this gauge
\begin{align}\label{darboux}
    F & = F_0(x+y,p)\,, 
    & \xi & = \xi_0(x+y,p)\,,
    & w & = w_0(x+y,p)\,,
\end{align}
and the residual gauge symmetry reproduces the off-shell Segal system
in terms of the initial data $F_0$, $\xi_0$ and $w_0$ (which are the $y$-independent piece of these $0$-forms). Note that $\xi$ is unconstrained
in \eqref{gs}, but the requirement to maintain Segal's gauge implies
$d\xi+\tfrac1\hbar\,[A,\xi]_\ast=0$, i.e. the parameters of the residual transformations are covariantly constant. 

Although equations \eqref{parentEOM} 
and \eqref{covconstw} involve space-time derivatives, they are equivalent to algebraic equations as they allow one
to reconstruct a solution from the initial data at $y=0$. Indeed,
as it is easy to observe in Segal's gauge, the initial data for $F$, $\xi$
and $w$ is given by arbitrary functions of $x$ and $p$. The fiber-wise
Moyal--Weyl star-product in $y-p$ space induces the one on $x$ and $p$,
and we can recover all formulas from Section \bref{sec:Segal}.

The notation suggests that $\xi$ is responsible for a covariantized version
of (higher spin) diffeomorphisms and $w$ is responsible for (higher spin)
Weyl transformations. A special feature  of this formulation is that the actual Segal gauge transformations
\eqref{Segalgauge} are associated with covariantly constant $\xi$ and $w$.
Parameters $w$ of (higher spin) Weyl symmetry are always constrained
by \eqref{covconstw} in order for the transformed $F$ to obey \eqref{parentEOM}.
Therefore, Segal gauge transformations should be associated with 
\begin{align}
    \delta_\xi A & =0\,,
    & \delta_{\xi,w} F&=\tfrac{1}{\hbar}\qcommut{F}{\xi}+\scommut{F}{w}\,,
\end{align}
where $\xi$ and $w$ obey
\begin{align}
    d\xi + \tfrac{1}{\hbar}\qcommut{A}{\xi} & = 0\,,
    & dw + \tfrac{1}{\hbar}\qcommut{A}{w} & = 0\,.
\end{align}
From the point of view of physical fields hidden in $A$ and $F$
(as the coefficients of their power series expansion in $y$ and $p$),
gauge transformations with unconstrained parameters $\xi$ represent
field redefinitions, which we discuss later in Section \bref{sec:gauges}. 
Further reduction of symmetries is possible if we consider
those that preserve a given background $A=A^{(0)}$ and $F=F^{(0)}$.
This way we recover the higher spin algebra $\hs(\square \phi)$
for $F^{(0)} = \tfrac12\,p^2 \equiv \tfrac12\,\eta^{ab}\,p_a p_b$
and $A^{(0)}=dx^\mu \delta^a_\mu \,p_a$. 

The advantage of the parent reformulation  is that it is globally well-defined for any space-time  manifold $\manX$.\footnote{Of course, the existence of Lorentzian metric imposes some restrictions on the topology of $\manX$.} In particular, it is manifestly coordinate independent and does not require any predefined background geometrical structures. 
In order to illustrate this property, let us mention that the system encodes a particular star product on $T^*\manX$, which is still determined by an affine connection on $\manX$ (along with extra structures, in general) hidden  in the field $A$, but now $A$ is a part of the field content and not a  predefined background field.

The parent Segal system allows us
to parameterize covariant derivatives of the physical fields hidden
in the initial data $F\rvert_{y=0}$ and $ A\rvert_{y=0}$ as auxiliary fields inside $F$ and $A$ that are reconstructed by solving \eqref{parentEOM}. Since the parent
Segal system gauges the higher-spin diffeomorphisms $\xi$ of the off-shell Segal system,
\eqref{parentEOM} encodes fully higher spin covariant derivatives of
the initial data, i.e., roughly speaking, it defines for us a higher spin covariant derivative $d+\tfrac1\hbar\,[A,\bullet]$. In order to construct an action,
we will also need an appropriate higher spin covariant version
of the measure $\sqrt{g}$, which is related to the quantum trace discussed below.

%%%%%%%%%%%%%%%%%%%%%%%%%%%%%%%%%%%%%%%%%%%%%%%%%%%%%%%%%%%%%
\subsection{Covariant action}
\label{sec:ourcoolAction}
%%%%%%%%%%%%%%%%%%%%%%%%%%%%%%%%%%%%%%%%%%%%%%%%%%%%%%%%%%%%%
What we are after is an invariant definition of Segal system,
which does not depend on the particular choice of coordinates
and/or predefined geometric structures such as an affine
connection determining the star-product. The first step is
to reformulate the system in the parent form~\eqref{parentEOM}-\eqref{covconstw}. In the second step,
we define an action in terms of the parent fields using
a version of the invariant trace proposed by Feigin, Felder 
and Shoikhet \cite[Sec. 4]{Feigin2005}. More specifically,
their construction allows one to define a trace over
the algebra of functions of a symplectic manifold endowed
with a star-product obtained using Fedosov quantization,
in an invariant way. In the case of the flat symplectic manifold
$\mathbb{R}^{2n}$, the trace reduces to the usual integral
of the phase space, used in Segal's formulation of CHS gravity. 

The core of the construction is the Hochschild $2n$-cocycle
with values in $\cA_{2n}^*$,
\begin{equation}
    \Phi: \underbrace{\cA_{2n} \tensor \ldots \tensor \cA_{2n}}_{2n+1\,\text{times}}
    \longrightarrow \fC
\end{equation}
of the Weyl algebra $\cA_{2n}$ generated by formal power series 
in $y^a$ and polynomials in $p_a$,
see Appendix~\bref{app:modification_FFS} for further details.
Using this cocycle, one can build a reduced polylinear map
\begin{equation}
    \mu:\cA_{2n}^{\tensor (n+1)} \longrightarrow \fC[p]\,,
\end{equation}
defined as\footnote{Another option is to plug in $dp_ay^a$ in the last slots and get a cocycle with values in top-forms of $p$-fiber which is a natural integration object over fiber (see also Appendix \bref{app:Fedosov}).}
\begin{equation}
    \mu(a_{0}|a_1,\ldots,a_n) = [\Phi](T_{p'}  a_{0};T_{p'} a_1,\ldots,
    T_{p'} a_n,y^{b_1},\ldots,y^{b_n})\epsilon_{b_1\ldots  b_n}|_{p'_a=p_a}\,,
\end{equation}
where $T_{p'}a(y,p)=a(y,p+p')$ and $[\Phi]$ denotes the antisymmetrization
of $\Phi$ in its $2n$ arguments, i.e. $[\Phi]$ is the associated
Chevalley--Eilenberg $2n$-cocycle with values in $\cA_{2n}^*$.

Let us list here the algebraic properties of the map $\mu$:
\begin{enumerate}[label=$(\roman*)$]
\item Total antisymmetry in its $n$ last arguments,
    \begin{equation}        \mu(a_0|a_{\sigma_1},\dots,a_{\sigma_n})
        = (-1)^{|\sigma|}\,\mu(a_0|a_1,\dots,a_n)\,,
    \end{equation}
    for any elements $a_0,a_1,\dots,a_n \in \cA_{2n}$ 
    of the Weyl algebra and any permutation $\sigma\in\cS_n$;
\item\label{item:normalization} The normalization condition,
    \begin{equation}\label{eq:normalization}
        \mu(f;p_{a_1}\,\dots,p_{a_n})
        = \tfrac1{n!}\,\epsilon_{a_1 \ldots a_n} f\rvert_{y=0}\,,
    \end{equation}
    for any $f \in \cA_{2n}$;
    \item\label{item:cocycle} The `cocycle condition', modulo
    total derivative in $p$-space,
    \begin{eqnarray}
        && \sum_{i=0}^n\,(-1)^i\,
        \mu([a_{-1},a_i]_\ast;a_0,\dots,\hat a_i,\dots,a_n) \\
        && \ + \sum_{0 \leq i<j \leq n}
        (-1)^{i+j}\,\mu(a_{-1};[a_i,a_j]_\ast,
        a_0,\dots,\hat a_i,\dots,\hat a_j,\dots,a_n)
        = \tfrac{\partial}{\partial p_a}
        \varphi_a(a_{-1};a_0,\dots,a_n)\,,
        \nonumber
    \end{eqnarray}
    for some $\varphi_a(a_{-1}|a_0,\dots,a_n) \in \C[p]$
    (see Appendix \bref{app:modification_FFS} for more details);
\item\label{item:sp2n} The $\mathfrak{sp}(2n,\R)$-invariance,
    \begin{equation}
        \mu(a;-,\dots,-) = 0 = \mu(-;-,\dots,a,\dots,-)\,,
    \end{equation}
    for any element
    $a \in \mathfrak{sp}(2n,\R) \subset \cA_{2n}$.\footnote{Recall
    that the Lie algebra $\mathfrak{sp}(2n,\R)$ is embedded in the Weyl
    algebra $\cA_{2n}$ as the subspace of quadratic elements.}
\end{enumerate}

The above structure allows us to define an action principle
as follows. Let $l_*(F)$ be a star-product function (see Appendix
\bref{app:starHeaviside} for more details), and consider
the following functional
\begin{equation}\label{eq:cov-action}
    S[A,F] = \int_\manX\int_{p-\text{fiber}}
    \mu\big(l_*(F);A, \dots, A\big)\,,
\end{equation}
where the fields $A$ and $F$ are subject to the off-shell
constraints \eqref{parentEOM}. This action is invariant
under the gauge transformations generated by $\xi$,
up to boundary terms: indeed, upon using the property
\ref{item:cocycle} of the map $\mu$, together with
the flatness and covariant constancy of $A$ and $F$
respectively, one can show that (see corollary \ref{cor:mu} in Appendix
\bref{app:modification_FFS})
\begin{equation}
    \delta_\xi\mu\big(l_\ast(F);A,\dots,A\big) \propto d(\dots)
    + \tfrac{\partial}{\partial p_a} (\dots)_a\,,
\end{equation}
i.e. the integrand of \eqref{eq:cov-action} is gauge
invariant up to a total derivative. Note that the particular
form of $l_\ast$ does not matter for this property,
what is important is that $l_\ast(F)$ is covariantly constant
with respect to the connection $A$, which is the case
since we assume $F$ to be covariantly constant.

The choice of an appropriate star-function $l_\ast$ is however
crucial to ensure the invariance of the above action under
higher spin Weyl transformations, i.e. gauge transformations
generated by $w$. To see that, let us first point out that
the action \eqref{eq:cov-action} can be interpreted as
a trace. Indeed, interpreting $A$ as a Fedosov connection the space $\cS(\manX)$ of functions on $T^*\manX$  can be endowed
with a star-product, via
\begin{equation}
    f \star g := (F \ast G)\rvert_{y=0}\,,
\end{equation}
where $F=F(f)$ and $G=G(g)$ are the unique covariantly constant
sections such that $F\rvert_{y=0}=f$ and $G\rvert_{y=0}=g$
(see Appendix \bref{app:proof}).
Then for any $f \in \cS(\manX)$ of compact support, the operation
\begin{equation}\label{muformula}
    \tr_A(f) := \int_\manX\int_{p-\text{fiber}}
    \mu(F;A,\dots,A)\,,
\end{equation}
defines a trace, in the sense that it verifies
\begin{equation}
    \tr_A(f \star g) = \tr_A(g \star f)\,,
\end{equation}
for any other $g \in \cS(\manX)$. This cyclicity property
also follows directly from the flatness of $A$, the covariant
constancy of $F$ and $G$, and the cocycle condition
obeyed by $\mu$, upon discarding total derivative terms.
The action \eqref{eq:cov-action} can then be re-written as
\begin{equation}
    S[A,F] = \tr_A\big(l_\star(f)\big)\,,
\end{equation}
with $f=F\rvert_{y=0}$, and its variation under higher-spin Weyl 
transformations reads
\begin{eqnarray}
    \delta_w S = 2\,\tr_A\big(l'_\star(f) \star f \star w_0\big)
    = \int_\manX\int_{p-\text{fiber}}
    2\,\mu\big(l'_\ast(F) \ast F \ast w; A, \dots, A\big)\,,
\end{eqnarray}
where $w_0 = w\rvert_{y=0}$. As in the Segal case,
one can see that the choice $l_\ast(F) = \Theta_\ast(F)$ --- the Heaviside
function, guarantees the invariance of the action under
HS Weyl transformations,
since $l'_\ast(x) \ast x = \delta_\ast(x) \ast x = 0$
(at least formally, see Appendix \bref{app:starHeaviside}
for more details). On top of that, this choice also
implies that our action reduces to Segal's around flat space.
Indeed, for $\manX=\R^n$, the choice $A^{(0)}=dx^\mu \delta_\mu^a\,p_a$ can be
made globally, so that the action reduces to
\begin{equation}
    S[f]={S[A=A^{(0)},F=F(f)]}
     = \int_{\R^{2n}} \Theta_\star(f)\,,
\end{equation}
upon using the normalisation property \ref{item:normalization}
of $\mu$, and the fact that the star-product $\star$
simply becomes the Moyal--Weyl star-product in $x$ and $p$
for this particular choice of connection $A=A^{(0)}$.

The above analysis implies
that the system is gauge invariant under the transformations
generated by the gauge parameters $\xi$ and $w$. Let us dwell on
the interpretation of the system: action~\eqref{eq:cov-action} 
is understood as a functional defined on the space of solutions
of the off-shell system \eqref{parentEOM}. As we are going
to see in the next subsection, using the gauge freedom
one can set $A$ to be a fixed connection while solutions
for $F$ are 1:1 with the unconstrained configurations
$f(x,p)=F|_{y=0}$. Consequently, the functional
\eqref{eq:cov-action} gives a globally well-defined action
on the configuration space of unconstrained $f(x,p)$.

%%%%%%%%%%%%%%%%%%%%%%%%%%%%%%%%%%%%%%%%%%%%%%%%%%%%%%%%%%%%%%%%%%%%%
\section{Gauge conditions, field redefinitions, and background fields}
\label{sec:gauges}
%%%%%%%%%%%%%%%%%%%%%%%%%%%%%%%%%%%%%%%%%%%%%%%%%%%%%%%%%%%%%%%%%%%%%
The action \eqref{eq:cov-action} supplemented with the off-shell constraints  \eqref{parentEOM} and gauge transformations~\eqref{gs} and \eqref{covconstw}
of the parent Segal system defines the action of CHS gravity in a covariant and
coordinate-independent way. However, in this formulation the system involves an overcomplete set of fields, which effectively reduces to the minimal one only upon taking into account off-shell constraints and algebraic gauge transformations. Below we discuss this procedure in more details, and identify several useful gauge conditions.

%**********************%
\subsection{Segal gauge}
%**********************%
Our first task is to demonstrate that locally any solution
to the zero-curvature equation \eqref{parentEOM} is equivalent to 
one where $A=A^{(0)}$ with some fixed $A^{(0)}$. Recall that
we consider connections $A$ whose piece linear in $p$
is invertible, i.e. $e_\mu^a$ entering $dx^\mu e_\mu^a(x) p_a$
is invertible. In other words, we want to prove that locally all flat
connection on the Weyl bundle with an invertible $e^a_\mu$ belong to the same gauge orbit. The parent version
of HS diffeomorphisms act on $A$ and $F$ as
\begin{equation}
    \label{gs-finite}
    \begin{aligned}
        A^\prime & = e_\ast^{-\frac{\lambda}{\hbar}} \ast
        (\hbar d+A) \ast e_\ast^{\frac{\lambda}{\hbar}}\,,
        & \qquad 
        F^\prime & = e_*^{-\frac{\lambda}{\hbar}} \ast F \ast 
        e_\ast^{\frac{\lambda}{\hbar}} \,,\\
        \qquad
        \delta_\lambda A
        &=d\lambda+\ffrac{1}{\hbar}\qcommut{A}{\lambda}\,,
        & \qquad
        \delta_\lambda F
        & = \ffrac{1}{\hbar}\qcommut{F}{\lambda}\,, 
    \end{aligned}
\end{equation}
where in the second line we list the infinitesimal version. 
In the present local analysis, we take $A^{(0)}=dx^\mu\delta_\mu^a p_a$. It is convenient to denote space-time coordinates by $x^a$ so that $A^{(0)}=dx^a p_a$.

Using the above gauge transformations with $\lambda$
of the form $\lambda_a(x) y^a$, one can set to zero
the $y,p$-independent term in $A$.
Then $\lambda=\lambda^b_a y^a p_b$ allows us to set $e^a_b=\delta^a_b$.
Furthermore, the flatness condition for $A$ sets to zero
the antisymmetric part of $h_{ab}$ entering $A$ as
$dx^a h_{ab}y^b$ in $A$. The symmetric part is then set to zero
by gauge transformations with $\lambda$ of the form
$\lambda_{ab}y^a y^b$ so that one can assume that
$A=dx^a p_a+\text{terms of higher order in $y,p$}$.

Suppose that $A^\prime$ is another flat connection
with invertible $e^a_b$.  As above we can also assume that
it starts with $dx^ap_a$. Introduce the following degree:
\begin{equation}
\label{Fedosov-degree}
    \deg(y) = 1 = \deg(p)\,,
    \qquad 
    \deg(\hbar) = 2\,,
\end{equation}
which is precisely the degree used in Fedosov quantization,
and hence we will refer to it from now on as the Fedosov degree.
Expanding $A$ and $A^\prime$ according to this degree
as $A=A_{(1)}+A_{(2)}+\ldots$, and similarly for $A'$,
one has $A_{(1)}=A^\prime_{(1)}$. 

We then continue by induction. To this end, let us assume
that $A_{(l)}=A^\prime_{(l)}$ for all $l\leq k$.
The zero-curvature equations for $A$ and $A'$ imply
\begin{equation}
\delta (A^\prime_{(k+1)}-A_{(k+1)})=0\,, \qquad \delta\equiv-\frac{1}{\hbar}\qcommut{A_{(1)}}{\cdot}=dx^a\dl{y^a}\,,
\end{equation}
where following Fedosov we have introduced a nilpotent operator $\delta$. Because the cohomology of $\delta$ is trivial in nonvanishing form-degree, it follows $A_{(k+1)}^\prime-A_{(k+1)}=\delta \lambda_{(k+2)}$ for some $\lambda_{(k+2)}$. Applying gauge transformation~\eqref{gs-finite} with $\lambda=\lambda_{(k+2)}$ to $A^\prime$ one finds $A-A^\prime$ is of degree $k+2$ or higher. In particular, taking $A=dx^ap_a$, one finds that this gauge is locally reachable. 

Despite this gauge being reachable
only locally, it is very instructive. In particular,
it is clear that the covariant constancy condition 
$dF+\tfrac1\hbar\,\qcommut{A^{(0)}}{F}=0$
has a unique solution $F=f(x+y,p)$ satisfying
$F|_{y=0}=f$ for any unconstrained $f=f(x,p)$.

%*****************************%
\subsection{Metric-like gauges}
%*****************************%
Segal's gauge just discussed is the simplest example
of a gauge where all the independent fields are contained
in $F$, while $A$ is set to a background value. We refer
to such gauges as metric-like ones. 

A class of globally well-defined connections
on the Weyl bundle can be constructed starting with
a given torsion-free affine connection.
More specifically, let 
\begin{equation}
    \varpi = e^a\,\PPP_a + \Gamma^a{}_b\,\TTT^b{}_a\,,
    \qquad
    \PPP_a = p_a\,,
    \qquad 
    \TTT^a{}_b := y^a\,p_b\,,
\end{equation}
verify
\begin{equation}
    de^a + \Gamma^a{}_b\,e^b = 0\,,
    \end{equation}
so that its curvature is given by
\begin{equation}
    d\varpi+\tfrac1{2\hbar}\,[\varpi,\varpi]_\ast
    = R^a{}_b\,\TTT^b{}_a\,.
\end{equation}
It is a standard statement~\cite{Fedosov:1994zz} (see also \cite{Grigoriev:2006tt,Grigoriev:2016bzl} for precisely this setup and Appendix~\bref{app:proof} for a proof of a more general statement) that such a connection has a unique  completion $A$
such that $A_{(0)}=0, A_{(1)}+A_{(2)}=\varpi$ and $hA_{(l)}=0$ for $l>2$,  where we again use the decomposition in Fedosov degree~\eqref{Fedosov-degree}.
Here $h$ is a contracting homotopy for $\delta$, given by
\begin{equation}
\label{eq:contracting}
    h = \tfrac1N\,y^a\,e^\mu_a\,\imath_{\partial_\mu}\,,
\end{equation}
where $\imath_{\dd_\mu}$ denotes the interior product by
the vector field $\dd_\mu$, and $N$ is the operator counting
the sum of the form degree as well as the degree in $y$.
Note also that the completion is such that all $A_{(l)}$
are linear in $p$. The first few orders read as
\begin{equation}
    \begin{aligned}
        A & = e^a\,\PPP_a + \Gamma^a{}_b\,\TTT^b{}_a
        -e^a\,\Big(\tfrac13\,
        R_{ab}{}^d{}_c\,y^b y^c\,p_d
        +\tfrac1{12}\,\nabla_b R_{ac}{}^e{}_d\,y^b y^c y^d\,p_e \\
        & \hspace{120pt} + \big[\tfrac1{60}\,\nabla_b \nabla_c R_{ad}{}^f{}_e
        + \tfrac2{45}\,R_{ab}{}^g{}_c\,R_{de}{}^f{}_g\big]\,y^b y^c y^d y^e\,p_f + \dots\Big)\,,
    \end{aligned}
\end{equation}
where the $\dots$ denote corrections of higher order in $y$,
but do not contain any terms in $\hbar$
(this is due to the fact that $\varpi$ is linear in $p$,
see Appendix \bref{app:proof}). Similarly,
the first few orders of a covariantly constant section $F$
such that $F\rvert_{y=0}=f(x,p)$ are given by
\begin{equation}
    F = f + y^a\,\nabla_a f + \tfrac12\,y^a y^b\,
    (\nabla_a \nabla_b + \tfrac13\,R_{ab}{}^c{}_d\,
    p_c\,\tfrac{\partial}{\partial p_d})\,f + \dots\,,
\end{equation}
where the $\dots$ represent higher order corrections in both $y$
and $\hbar$.

It can be useful to take $\varpi$ to be a metric-compatible
connection. In this case, $A$ constructed above and $F$ determined by
$f=\tfrac12\,\eta^{ab}p_ap_b$, where $\eta$ is a Minkowski metric,
describe a gravitational background. In particular, given a frame
$e$, the flat connection $A$ is entirely determined by the metric
$g_{\mu\nu}=\eta_{ab} e^a_\mu e^b_\nu$.

The gauge just constructed can be considered as a covariantized
and globally well-defined version of the Segal gauge. A reason
to call it metric-like is that the independent fields encoded
in $f(x,p)$ are totally symmetric tensors (a covariantized version
of Fradkin--Tseytlin fields). This gauge is useful in the analysis
of the propagation of CHS fields on gravitational backgrounds, and it was already 
employed in this context in~\cite{Grigoriev:2016bzl}.

By employing a fixed metric-like gauge, the covariantized Segal action becomes a functional of the metric-like CHS fields encoded in $f(x,p)$. Because this action is gauge-invariant, the gauge-fixed actions corresponding to gauge-equivalent background connections $A$ and $A^\prime$ should be related by a field redefinition.

Note also that by a suitable gauge transformation one can set $e^a_\mu=\delta^a_\mu$, i.e. pick a coordinate local frame. In this case, the flat connection reads
\begin{equation}
    A = dx^\mu\,(p_\mu + \Gamma_{\mu\nu}^\lambda\,y^\nu\,p_\lambda
    -\tfrac13\,R_{\mu\nu}{}^\lambda{}_\sigma\,
    y^\nu y^\sigma\,p_\lambda + \dots)\,.
\end{equation}

Let us also mention that, due to the $\mathfrak{sp}(2n,\R)$-invariance
of cocycle $\mu$, \eqref{muformula}, the connection
$\Gamma$ never appears alone in the expression of the action.
Indeed, it is the component of $A$ along a quadratic element
of the Weyl algebra, and hence belongs to its $\mathfrak{sp}(2n,\R)$
subalgebra. In other words, $\Gamma$ will appear in the final
action only through the covariant derivative $\nabla$,
or its curvature $R$.

The above construction of a flat connection starting from
a curved one is an instance of the slightly more general mechanism
of connection flattening in the Weyl algebra, which is
encompassed in the following proposition.
\begin{proposition}\label{prop:flattening}
    Let $\cD$ be a connection on the Weyl bundle $W(\manX)$,
    whose connection $1$-form is $\cA_{2n}$-valued and
    acts in the adjoint. Let in addition the Fedosov
    degree $1$ piece of $\cD$ be an invertible vielbein.
    Then there exists a $1$-form
    $w\in \Omega^1(\manX)\tensor \Gamma(\hat W(\manX))$ 
    such that $\cD + \tfrac1\hbar\,[w,-]_\ast$
    is a flat connection. Moreover,
    any $f \in \Gamma\big(S(T\manX)\big)$ has a unique
    completion to a section $F \in \Gamma(W(\manX))$
    such that $\cD F + \tfrac1\hbar\,[w,F]_\ast=0$
    and $F|_{y=0}=f$. We refer to $F$ as to a covariantly-constant
    lift of $f$.
\end{proposition}
Say that the connection on the Weyl bundle is (locally)
given by $\cD=d+\tfrac1\hbar\,[\varpi,-]_\ast$,
then the flat connection is given by $d+\tfrac1\hbar\,[A,-]_\ast$
with $A=\varpi+w$. For instance, in the previous example
we considered an affine connection encoded by
the $\varpi=e^a\,\PPP_a+\Gamma^a{}_b\,\TTT^b{}_a$ valued
in $\mathfrak{igl}(n,\R)$, but in principle, one could
consider more general, nonlinear connections as a starting
point. The proof of this statement is a minor modification
of the proof that any symplectic connection can be lifted
to a flat connection on the Weyl bundle of a symplectic
manifold \cite[Th. 3.2]{Fedosov:1994zz} (see also
\cite[Th. 2]{Dolgushev:2003fg}) and is therefore relegated
to Appendix \bref{app:proof}.

%***************************%
\subsection{Frame-like gauges}
\label{sec:frame-like-gauge}
%***************************%
In some sense, opposite to metric-like gauges are frame-like
ones. In these gauges, all the dynamical fields are contained
in the $A$-field and are interpreted as components
of a connection $1$-form. 

To see this, let us again use the Fedosov degree~\eqref{Fedosov-degree}
and start with a generic solution $(A,F)$ of the parent system
(recall that $e^a_\mu(x)$ in the term $dx^\mu e^a_\mu(x) p_a$
in $A$ is assumed invertible). In addition, we now also assume 
$g^{ab}(x)$ in the term $\tfrac12\,g^{ab}p_ap_b$ in $F$
to be invertible and to have Lorentzian signature.
Performing a finite gauge  transformation~\eqref{gs-finite}
with $\lambda$ of the form $h_a(x)y^a$, one can achieve
$F_{(1)}=f_a(x)y^a$, i.e. remove the term linear
in $p_a$.\footnote{Strictly speaking, this is true
under the assumption that all $s\neq 2$ fields are small
compared to $s=2$ field. Otherwise, the respective series may 
diverge. }
In so doing, the transformed $A$ gets in addition the nonvanishing $A_{(0)}=A_{(0)}(x)$ contribution. In the next step, we perform a gauge transformation
with $\lambda$ of the form $\lambda^a_b y^b p_a$
in order to set $F_{(2)}$ to be $\tfrac12\,\eta^{ab}p_a p_b$,
where $\eta^{ab}$ is the inverse of the standard Minkowski 
metric.

We then proceed by induction. Suppose that by further gauge
transformations we succeeded to set the $p$-dependent parts
of $F_{(l)}|_{y=0}$ to zero  for all $l \leq k$ save for
$l=2$. The $p$-dependent part of $F_{(k+1)}|_{y=0}$
can be then gauged away by parameters of the form 
$\sum_m\hbar^{m}y^a\lambda_a^{b_1 \ldots b_{k-2m}}p_{b_1}\ldots p_{b_{k-2m}}$
(i.e. of degree $k+1$ and linear in $y$).
By degree reasoning, such gauge transformations cannot affect
$F_{(l)}$ with $l\leq k+1$. At the same time, we can eliminate
the $p$-independent term in $F_{(k+1)}|_{y=0}$ by the leading
(in $\hbar$) term  of the gauge transformation. However,
in so doing one can get nonvanishing contributions
in $F_{(k+1)}$, which are proportional to $\hbar$.
The procedure can then be iterated,
giving a $p$-independent $F_{(k+1)}$. 

The induction then implies that the gauge where 
$F|_{y=0}=D(x)+\tfrac12\,\eta^{ab}p_ap_b$ is reachable.
In other words, the configuration of the parent system
is entirely determined by the configuration of the connection
$A$ (although, strictly speaking $D$ is not captured by $A$).
Indeed, Proposition~\bref{prop:flattening} implies that 
$D(x)+\tfrac12\,\eta^{ab}p_ap_b$ admits a unique covariantly 
constant lift $F$ satisfying $F|_{y=0}=D(x)+\tfrac12\,\eta^{ab}p_ap_b$.

A natural question is then which components of $A$
can be taken as independent fields.  It turns out that
the minimal set is given by the HS frame field encoded
in 
\begin{equation}
    E(p) = dx^\mu E_\mu(x,p)
    = dx^\mu\,(a_\mu(x)+ e_\mu^a p_a+\ldots)\,.
\end{equation}
For the covariance we also take a fixed torsion-free
Lorentz connection $\Gamma=dx^\mu\omega_\mu{}^a{}_b\,y^b p_a$.
Starting with
\begin{equation}
    \varpi = dx^\mu E_\mu(x,p) - dx^\mu (\partial_\mu a_\nu)
    e^\nu_ay^a+\Gamma\,,
    \qquad
    e^a_\mu e^\mu_b=\delta^a_b\,,
\end{equation}
a minor modification of the proof of the Proposition~\ref{prop:flattening}
allows one to construct a flat connection $A$ satisfying 
$A|_{y=0}=E(p)$ (and an extra condition,
see Appendix~\bref{app:proof}). An alternative proof, is given in Appendix~\bref{app:proof-HSframe}.

We have just seen that the HS frame field (whose configurations
are 1:1 with metric-like conformal higher spin fields if one takes the totally 
symmetric components of the HS frame) serve as the initial data
for the $A$-field. In particular, the covariantized 
action~\eqref{eq:cov-action} can be written in this gauge
as a functional of the HS frame field $E(p)$. We conclude
that this approach also generate a frame-like description
of the Segal action. 

It is worth mentioning that the parent formulation
was initially developed to explicitly relate metric-like
and frame-like formulations. In particular, the Lagrangian
version of the parent formalism~\cite{Grigoriev:2010ic,Grigoriev:2012xg}
allows one to systematically derive a frame-like formulation
starting from the metric-like one, so that it is not surprising
that these formulations are reproduced through different gauges
of the parent system. 

To illustrate the relation between the frame-like and
the metric-like gauges, let us fixe a gravitational background 
described by the frame $e^a$ and the Lorentz 
connection $\omega^{ab}$, so that the metric is 
$g_{\mu\nu}=e^a_\mu e^b_\nu \eta_{ab}$, and consider
a linearized spin-$s$ field $\Phi^{a(s)}(x)$
on this background described in the metric-like gauge.
We restrict ourselves to a single CHS field in $F$
since the argument is about free fields for simplicity.
More precisely, the configuration for $A$ and $F$ reads as 
\begin{align}
    A & = e^a\,\PPP_a + \tfrac12\,\omega^{ab}\,\LLL_{ab}
    + \dots\,,
    & F & = \tfrac12\,p^2+ \dots + \Phi^{a(s)}\, p_a \dots p_a
    + \dots\,
\end{align}
where $\ldots$ denotes the terms that complete
the initial data of $A$ and $F$ to a solution of the parent
system \eqref{parentEOM}. Now, if we perform a gauge 
transformation with 
\begin{align}
    \xi & = \tfrac12\,\Phi^{a(s)}\, y_a p_a \dots p_a
    + \dots\,,
\end{align}
we get as a result
\begin{align}
    A & = e^a\,\PPP_a + \tfrac12\,\omega^{ab}\,\LLL_{ab}
    - \tfrac12\,e_a\,\Phi^{a b(s-1)}\, p_b \dots p_b
    + \tfrac12\,\nabla\Phi^{a(s)}\,y_a p_a \dots p_a + \dots\,,
\end{align}
and
\begin{equation}
    F = \tfrac12\,p^2 + \dots\,,
\end{equation}
i.e. we have moved the linearized spin-$s$ field from $F$
to $A$. In the last two components of $A$, we see the CHS
vielbein $e^{a(s-1)}=e_m\Phi^{m a(s-1)}$ in a particular
gauge followed by its first auxiliary field that is expressed
in terms of the first derivative of $\Phi^{a(s)}$.   

Let us finally mention that in showing the existence
of the metric-like and the frame-like gauges, we only made use 
of HS diffeomorphisms. It follows that the analogous gauges 
are reachable in the version of the parent 
system~\cite{Vasiliev:2005zu,Grigoriev:2006tt,Grigoriev:2012xg}
with HS Weyl transformations dropped, which describes 
nonlinear gauge transformations of the off-shell
massless HS fields with the trace constraint relaxed. Let us note that massless higher spin fields within a similar framework were discussed in \cite{Ponomarev:2013mqa}.

%***********************************%
\subsection{Conformal geometry gauge}
%***********************************%
We now get back to metric-like gauges and demonstrate that
with a suitable choice of $A$ one can make the underlying 
conformal geometry manifest.

The idea of this approach is to observe that the conformal
algebra $\mathfrak{so}(n,2)$ can be identified as a Lie subalgebra
of $\cA_{2n}$, and more specifically of its subalgebra of elements 
linear in $p$. In the standard basis, the commutation relations
read as
\begin{subequations}
    \begin{align}
        [\DDD,\PPP^a] & = +\PPP^a\,,
        & [\LLL^{ab},\PPP^c]
        & = \PPP^a\eta^{bc} - \PPP^b\eta^{ac}\,,\\
        [\DDD,\KKK^a] & = -\KKK^a\,, 
        & [\LLL^{ab},\KKK^c] & 
        = \KKK^a\eta^{bc} - \KKK^b\eta^{ac}\,,\\
        [\KKK^a,\PPP^b] & = \eta^{ab}\DDD - \LLL^{ab}\,,
        & [\LLL^{ab},\LLL^{cd}] & = \LLL^{ad}\eta^{bc}
        + \text{three more}\,.
    \end{align}
\end{subequations}
with $\PPP_a$, $\LLL_{ab}$, $\DDD$ and $\KKK_a$ the generators
of translations, the Lorentz transformations, dilation and
special conformal transformations respectively.
They can be represented in $\cA_{2n}$ as 
\begin{equation}\label{conformalGens}
    \PPP_a = p_a\,,
    \qquad 
    \LLL_{ab} = 2\,p_{[a}\,y_{b]}\,,
    \qquad
    \DDD = y^a\,p_a + \Delta\,,
    \qquad
    \KKK_a = y_a\,(y \cdot p + \Delta) - \tfrac12\,y^2\,p_a\,,
\end{equation}
where $\Delta \in \R$ is any real number for the moment.

A generic connection $\varpi$ valued in the conformal algebra,
\begin{equation}
    \varpi = e^a\,\PPP_a + \tfrac12\,\omega^{ab}\,\LLL_{ab}
    + b\,\DDD + f^a\,\KKK_a\,,
\end{equation}
can be put in a simpler
form by gauge fixing its component $b$ to zero, imposing
that it is torsionless (so that the spin-connection is
expressed in terms of the vielbein) and taking $f^a$
to be the Schouten $1$-form\footnote{This last condition
follows from imposing that the curvature of $\varpi$
along the Lorentz generators $F[\varpi]^{ab}$ is traceless
in the sense that $F[\varpi]_{\mu\nu}^{ab}\,e^\nu_b=0$.
This connection is called the normal Cartan connection
\cite[Def. 1.6.7 or 3.1.12]{Cap2009}, see also
\cite[Sec. 2]{Joung:2021bhf}
and \cite[Sec. 3.2]{Dneprov:2022jyn} for more details.}
$P^a$, whose components read
\begin{equation}
    P_\mu{}^a := \tfrac1{n-2}\,\big(R_\mu{}^a
    - \tfrac1{2(n-1)}\,e_\mu^a\,R\big)\,,
\end{equation}
where $R_\mu{}^a = R_{\mu\nu}^{ab}\,e^\nu_b$
and $R=R_\mu^a\,e^\mu_a$, with $R_{\mu\nu}^{ab}$ the Riemann
curvature of $\omega$. The curvature of this gauge fixed version
of the conformal connection,
\begin{equation}
    \varpi = e^a\,\PPP_a + \tfrac12\,\omega^{ab}\,\LLL_{ab}
    + P^a\,\KKK_a\,,
\end{equation}
takes the simple form
\begin{equation}
    d\varpi + \tfrac1{2\hbar}\,[\varpi,\varpi]_\ast
    = \tfrac12\,C^{ab}\,\LLL_{ab} + (\nabla P^a)\,\KKK_a\,,
\end{equation}
where $C^{ab} := R^{ab} - 2\,e^{[a}P^{b]}$ is the Weyl $2$-form.

We can now apply the flattening procedure of Proposition~\bref{prop:flattening} to construct
a flat connection $A$ starting from the previously
described conformal connection $\varpi$. The first few orders
of $A$ are given by
\begin{equation}
    \begin{aligned}
        A & = e^a\,\PPP_a + \tfrac12\,\omega^{ab}\,\LLL_{ab}
        + P^a\,\KKK_a + \tfrac13\,e^a\,C_{abc}{}^d\,y^b y^c\,p_d\\
        & \hspace{100pt}
        + e^a\,\big(\tfrac12\,\nabla_{[b} S_{a]|cd}{}^e
        + \tfrac1{12}\,\nabla_b C_{acd}{}^e\big)\,y^b y^c y^d\,p_e
        + \dots\,,
    \end{aligned}
\end{equation}
where we introduced the tensor
\begin{equation}
    S_{a|bc}{}^d := P_{ae}\,(\delta^e_{(b}\delta^d_{c)}
    -\tfrac12\,\eta^{ed}\eta_{bc})\,,
\end{equation}
and by construction, the higher orders terms in $A$
will be contraction of the covariant derivative of the Weyl
tensor $C$ and the Schouten tensor $P$.
Having determined $A$, we can now turn our attention to
the $0$-form $F$, which can be constructed as a covariantly
constant lift of an unconstrained $f(x,p)$. The first few orders
of $F$ are given by
\begin{equation}
    F = f + y^a\,\nabla_a f + \tfrac12\,y^a y^b\,
    (\nabla_a \nabla_b + [\tfrac13\,C_{adb}{}^c + 2\,S_{a|bd}{}^c]\,
    p_c\,\tfrac{\partial}{\partial p_d})\,f + \dots\,,
\end{equation}
where the dots indicate corrections of higher order in $y$
and $\hbar$.\footnote{Note that $\hbar$ corrections will appear
in $F$ only if $f$ contains terms which are at least quadratic
in $p$.}

Now let us focus on the spin-$2$ case. That is,
we consider that the $0$-form $F$ is simply
the completion of $\tfrac12\,p^2=\tfrac12\,\eta^{ab}\,p_a p_b$,
\begin{equation}
    F = \tfrac12\,p^2 + y^a y^b\,(\tfrac16\,C_a{}^c{}_b{}^d
    + S_{a|b}{}^{cd})\,p_c p_d + \dots\,,
\end{equation}
into a covariantly constant section. With $A$ being
the completion of the normal Cartan connection into
a flat connection of the Weyl bundle, this gauge
is a frame-like one. The gauge transformations generated by
a parameter $w$ which is the lift of the Weyl parameter $\sigma$,
\begin{equation}
    w = \sigma + y^a\,\nabla_a \sigma
    + \tfrac12\,y^a y^b\,\nabla_a \nabla_b \sigma
    + \tfrac16\,y^a y^b y^c\,\big(\nabla_a \nabla_b \nabla_c
    + 2\,S_{a|bc}{}^d\,\nabla_d\big)\sigma + \dots\,,
\end{equation}
reads
\begin{equation}
    \delta_w F = p^2\,\sigma + p^2\,y^a\,\nabla_a\sigma
    + \cO(y^2)\,,
\end{equation}
thereby signaling $F$ does encode a conformal spin-$2$ 
field. Since the connection $A$ also contains a conformal 
connection, and therefore also describes a conformal spin-$2$ 
field, let us check that their gauge transformations
are compatible. By this, we mean to check whether
it is possible to find a gauge parameter $\xi$ such that
the frame-like gauge is preserved, i.e.
\begin{equation}
    \delta_{\xi,w} F = \cO(y^2)\,,
\end{equation}  
so that the connection $A$ is transformed, and only
the completion of $\tfrac12p^2$ in $F$, i.e. terms
of order $2$ or higher in $y$ and $\hbar$, are affected.
Inspecting the above equation at the first few orders,
one finds that $\xi$ should take the form
\begin{equation}
    \xi = \sigma\,\DDD + \partial_a\sigma\,\KKK^a\,,
\end{equation}
and hence it implements the usual gauge transformation
of the normal conformal connection.\footnote{The fact
that the coefficient of $\KKK_a$ in the gauge parameter
is the derivative of the Weyl parameter $\sigma$ is a consequence
of the fact that we have gauge fixed $b$ (the gauge field
associated with dilation) to zero earlier.}

%****************************%
\subsection{Higher spin gauge}
%****************************%
It is sometimes convenient to rearrange a theory in terms
of the symmetries of one of its (maximally symmetric) backgrounds.
In the CHS gravity case, this corresponds to
$F^{(0)}=\tfrac12\,p^2 \equiv \tfrac12\,\eta^{ab}\,p_a p_b$.
As it was already discussed in Section \bref{sec:Segal},
the global symmetry algebra of this background is exactly
the higher spin algebra $\hs(\square \phi)$ of higher symmetries
of Laplacian \cite{Eastwood:2002su,Segal:2002gd}
(more or less by definition). The latter also fixes the conformal
weight $\Delta$ in \eqref{conformalGens} accordingly. 

It turns out that one can reconstruct field configurations
in the frame-like gauge of Section~\bref{sec:frame-like-gauge}
in terms of a connection of the HS algebra. To this end, let us fix
an embedding of HS algebra as a subspace in $\cA_{2n}$,
together with a projection to this subspace 
and take as $\varpi$ a connection with values in the subspace.
It can be then lifted to a flat connection $A$ by applying 
Proposition~\ref{prop:flattening}. What is important is
that independent fields sit in the HS frame part (which is 
the $y$-independent part of $\varpi$) but the completion
does not affect this part and hence gives a particular lift
of the HS frame to a flat $A$. The $F$ field is then reconstructed
by a covariantly constant lift of $\tfrac12\,p^2+D(x)$.

In this way we conclude that we succeeded to parameterize solutions to the parent system in terms of a connection of HS algebra. Of course this parametrization is much more redundant than the one in terms of HS frame. Moreover, it is also not clear how to explicitly identify the HS algebra as a gauge algebra in this setup. It would also be very interesting to come up with an appropriate
higher spin extension of the normal Cartan connection.

%***************************************************%
\subsection{Gauge symmetries vs. field redefinitions}
%***************************************************%
Given that the action of CHS gravity $S[A,F]$ depends on two fields that are subject to constraints \eqref{parentEOM}, it feels necessary to dwell on possible interpretations of such an action.  As we have already discussed, one can make use of a metric-like gauge where the action becomes a functional of the initial data $f(x,p)$ only. Of course, it still depends on a fixed connection $A$ but it is considered as a parameter, or better, a background field. Thanks to the gauge invariance of the covariant action, the change of $A$ leads to a field redefinition in terms of $f(x,p)$. This should be compared with the standard background field method (for gravity), see e.g. \cite{Fradkin:1985am}. Note that a version of the background field method in precisely this context was used in~\cite{Grigoriev:2016bzl}.

\iffalse
First of all, it is always possible to roll back to Segal/Darboux gauge locally, which proves that the field content of our action is equivalent to that of Segal. We have also discussed above some other gauge choices that redistribute the degrees of freedom between $A$ and $F$. This should be compared with the standard background field method (for gravity), see e.g. \cite{Fradkin:1985am}. For example, given a (conformal) gravity action $S[g]$ one can pick a background $g_0$, which does not have to solve the classical equations of motion, and consider the path integral with $S[g_0+g]$. Essentially this is what happens in various gauges of the parent Segal system, where 'different gauges' of this bigger (as compared to the off-shell Segal one \eqref{Segalgauge}) off-shell gauge theory correspond to different field frames/backgrounds. It is worth noting that $A$/$F$ favour frame-like/metric-like interpretations, respectively. For instance, the spin-two field is naturally described by the vielbein $e^a$ of $A$ and/or by the metric $g^{\mu\nu}$ of $F$, which makes the usual background field split $g_0+g$ look more complicated.  
\fi

If, on the contrary, we employ a frame-like gauge, where the parent field configuration is determined solely by a HS frame $E(p)$, the action becomes a functional of $E$ which is unconstrained. Alternatively, one can go for  a more redundant description in which $A$ is parametrized by a connection of the HS algebra.

For applications it may be useful to distinguish between three different types of gauge
symmetries of the parent Segal system:
\begin{enumerate}[label=$(\arabic*)$]
\item Those generated by parameters $\xi$ that are {\it not}
covariantly constant ($d\xi+\tfrac1\hbar[A,\xi]_\ast\neq0$)
correspond to field redefinitions: they allow one to move components
of $A$ and $F$ into one another (as illustrated at the end of Section \bref{sec:frame-like-gauge}). In fact, one has to consider the quotient of all $\xi$ by the covariantly constant ones;
\item Those generated by {\it covariantly constant} parameters $\xi$
and $w$ are a covariant version of Segal's original gauge symmetries,
in the sense that they affect only the $0$-form $F$ via the commutator
and anti-commutator respectively;
\item Those generated by covariantly constant gauge parameters
$\xi$ and $w$ and that also preserve a given vacuum $F^{(0)}$
correspond to global symmetries and define a higher spin algebra, $\hs(F^{(0)})$.
\end{enumerate}

%%%%%%%%%%%%%%%%%%%%%%%%%%%%%%%%%%%%%%%%%%%%%%%%%%%%%%%%%%%%%
\section{Conclusions and Discussion}
%%%%%%%%%%%%%%%%%%%%%%%%%%%%%%%%%%%%%%%%%%%%%%%%%%%%%%%%%%%%%

We constructed a covariant action for the simplest class
of conformal higher spin gravities, which can be associated 
with the free scalar conformal matter, $\square \phi=0$.  
CHS gravity is the theory of the background conformal fields
that couple to `single-trace' operators, or to put it simply,
bilinear operators $J_s=\phi\pl\dots\pl\phi+\dots$,
most of which are conserved (higher spin) tensors. 

The constructions of CHS gravity proposed by Tseytlin and Segal
are closely related and prove the theories to be well-defined
(at least in terms of a perturbative expansion around flat space
background). It goes without saying that theories
of gravity should admit manifestly covariant, coordinate- and
background-independent formulations. 
Addressing this question is the goal of the present paper.
It is quite amusing that the action of CHS gravity requires
such advanced constructions from deformation quantization as
Shoikhet--Tsygan--Kontsevich formality that gives a proper measure
for the invariant trace on the algebra of quantum observables,
using the Feigin--Felder--Shoikhet cocycle. Somewhat related links
to the same formality have already been observed 
\cite{Sharapov:2017yde}, in particular, for Chiral higher spin
gravity \cite{Sharapov:2022phg, Sharapov:2022wpz, Sharapov:2022nps}.

The induced action for CHS gravity can be derived
from a simple particle model, as mentioned in Segal's paper
\cite{Segal:2002gd} (see also \cite{Segal:2001di}),
or discussed in more details in \cite{Bonezzi:2017mwr}.
As it turns out, the action \eqref{eq:cov-action}
can also be obtained from a particle model. Indeed,
the main ingredient used to construct this action,
namely the Feigin--Felder--Shoikhet cocycle, admits
a representation as a correlation function in a particular
one-dimensional sigma-model, often called topological
quantum mechanics \cite{Grady:2015ica, Li:2018rnc}.
More specifically, this model is the simplest example
of AKSZ type \cite{Alexandrov:1995kv}, namely,
the $1d$ AKSZ model whose target space is the BFV--BRST
extended phase space of a constrained Hamiltonian 
systems.\footnote{Such AKSZ-like models were introduced 
in~\cite{Grigoriev:1999qz} and shown to produce
the BV formulation for the respective extended Hamiltonian
action. If the Hamiltonian is non-trivial, such models
have a slightly more general structure but in the case
at hand the Hamiltonian is trivial and the model is of AKSZ type.}
The underlying BFV--BRST system is precisely the BFV--BRST
reformulation of the particle model~\cite{Segal:2002gd}
underlying CHS theory and reviewed in Appendix~\bref{app:BRST}.
The Fedosov extension of the particle model again leads to
an extended AKSZ sigma model which produces the invariant trace
(and hence the FFS cocycle itself) as a correlation 
function~\cite{Grady:2011jc,Grady:2015ica}. Note that
the Fedosov-like extension itself can be understood as a passage
to the extended BFV--BRST system~\cite{Grigoriev:2000rn}.
Moreover, it is precisely the BFV--BRST system underlying
the parent reformulation~\cite{Barnich:2004cr,Grigoriev:2006tt}
we employ in this work. 

Extensions and generalizations of the present work should exist along
several lines: (a) one can choose different vacua $F^{(0)}$ for $F$, which is
equivalent to having the same type of matter, e.g. scalar, but with different
conformally-invariant equations, e.g. $\square^k \phi=0$, $k>1$
for the scalar matter corresponds to $F^{(0)}=(p^2)^k$; (b) one can choose different types of matter,
e.g. fermion $\psi$, or, more generally, a mixed-symmetry (spin-)tensor field, see \cite{Grigoriev:2018wrx} for the discussion of the CHS gravity based on the free fermion (called Type-B) and \cite{Bekaert:2009fg,Beccaria:2015uta} for further extensions;
(c) supersymmetric extensions should also be possible and be based
on the Clifford--Weyl algebra, see the recent \cite{Kuzenko:2022hdv,Kuzenko:2022qeq} for $\mathcal{N}=1$.  

We expect that the approach of this paper provides an efficient way
to attack some of the problems of conformal higher spin fields
that have been around for a while: whether conformal gravity is a consistent
truncation of CHS gravity \cite{Nutma:2014pua,Grigoriev:2016bzl,Beccaria:2017nco}?; what are the gravitational backgrounds
that admit free conformal higher spin fields
\cite{Nutma:2014pua, Grigoriev:2016bzl, Beccaria:2017nco, Kuzenko:2019ill, Kuzenko:2019eni, Kuzenko:2020jie, Kuzenko:2022hdv}?; the structure
of (higher spin) Weyl anomaly and, hence, the problem of quantum consistency of CHS gravity. 

It is well-known that the deformation quantization of a symplectic manifold $M$ up to a natural equivalence is in one-to-one with characteristic classes $\Omega[\hbar]=\Omega_0+\hbar \Omega_1+...$, where $\Omega_i\in H^2(M,\mathbb{C})$ and $\Omega_0$ is the class of the symplectic form. In the case of a cotangent bundle $\Omega_0$ is trivial. There is a simple deformation of the off-shell parent system \eqref{parentEOM}-\eqref{covconstw}, which allows one to put conformal higher spin fields on an external electromagnetic background, e.g. Dirac string. One needs a nontrivial class $H^2(\manX,\mathbb{C})$ 
of the base manifold $\manX$ itself, which can be added to the r.h.s.
\begin{align}
    dA + \tfrac1{2\hbar}\,\qcommut{A}{A} & = \Omega[\hbar]\,.
\end{align}
This is one simple generalization of CHS gravity that is not accessible in a local chart, where one can always impose Darboux coordinates. The possibility to add de Rham cohomology classes as deformations to higher spin systems seems to be a quite generic feature \cite{Boulanger:2015kfa,Sharapov:2021drr}. 

A closely related idea is that a general solution
$f=c_1\Theta(x)+c_2$ to $f'(x)x=0$ contains the constant term. 
Though it does not make any contribution in Darboux coordinates,
for a general compact symplectic manifold the quantity $\tr_A(1)$
leads to a particular case of Fedosov/Nest-Tsygan index theorem 
\cite{Fedosov1995, Nest1995, Nest1996}.
However, the cotangent bundle is non-compact. It would be interesting
to see if some of the index theorems admit higher spin extensions.
For example, Euler characteristic is the second conformal invariant
in $4d$ and appears on equal footing with the Weyl gravity action
in the studies of conformal anomalies. It remains an open question
whether it admits a higher spin extension and what is the interpretation 
of the corresponding topological invariant.      

As it was already pointed out in \cite{Segal:2002gd} and explored
in \cite{Joung:2015eny, Beccaria:2016syk}, a conformal higher spin gravity
can be coupled to the matter it originated from (via the effective action
approach or via the Segal approach). In the Segal approach the corresponding
coupling is simply
\begin{align*}
    S[H,\phi] & = S_{CHS}[H] + \langle \phi| \hat H \phi \rangle\,.
\end{align*}
It would be interesting to find its covariant extension along the lines
of the present paper.

Note that any conformal higher spin gravity can be truncated to its low spin
(not higher than spin-two) subsector by setting all higher spin fields to zero.
Nevertheless, a higher spin extension nicely fits the deformation quantization
framework: it is natural to consider all differential operators, which is 
an associative algebra (Weyl algebra), rather than to restrict to vector fields (see also \cite{Bekaert:2008sa, Bekaert:2021sfc} wherein similar ideas are advocated). The low spin subsector of CHS gravity allows us to make a bridge to conformal (super-)gravities.

Concerning the overlap between CHS gravities and conformal (super)gravities, let us mention a tremendous work
done in \cite{Butter:2016mtk,Butter:2019edc} that culminated in the complete
action of the maximal $\mathcal{N}=4$ conformal gauged supergravity. Curiously,
this action contains an arbitrary function of scalars, see \cite{Tseytlin:2017qfd} for the recent discussion. It is not clear
how to generate this function via the effective action idea since
there does not seem to be possible to introduce this ambiguity
into $\mathcal{N}=4$ SYM coupled to background conformal supergravity
fields \cite{deRoo:1985np}. It would be very interesting to see if the approach
advocated in this paper, i.e. higher spin geometry as deformation quantization,
can explain this ambiguity and extend it to higher spins.

Another interesting closely related class of theories are self-dual
truncations of CHS gravity, which admit a natural twistor space formulation \cite{Hahnel:2016ihf,Adamo:2016ple}. These theories are specific to four dimensions.
Since these theories are much simpler than the full CHS gravity, it can be instructive to see how their spacetime actions, which, in principle, are derivable from twistor space, can be formulated within our approach. Similarly, one can try to construct a covariant action for Chiral higher spin gravity \cite{Metsaev:1991mt,Metsaev:1991nb,Ponomarev:2016lrm,Skvortsov:2018jea}, which at present is available either in the light-cone gauge or for certain subsectors \cite{Ponomarev:2017nrr,Krasnov:2021nsq} only. 

Another interesting issue is whether the covariant action (and hence the underlying FFS cocycle) can be systematically derived within a purely field-theoretical framework. Indeed, starting from the Segal action one should be able to reconstruct its parent reformulation, and hence reconstruct an appropriate version of FFS cocycle, using the approach of~\cite{Grigoriev:2010ic,Grigoriev:2012xg}. A slightly alternative field-theoretical interpretation of the covariant action and its underlying cocycle is in terms of a suitable BRST-invariant presymplectic structure (see~\cite{Alkalaev:2013hta,Grigoriev:2020xec,Sharapov:2021drr,Dneprov:2022jyn} for more details on the presymplectic BV-AKSZ approach). This would signal an intriguing relation between the geometry of local gauge theories and algebraic structures underlying the deformation quantization.

Let us also point out that the way CHS gravity emerges in the Segal construction
is somewhat similar to the IKKT model based on a higher spin algebra
studied e.g. in \cite{Sperling:2017dts,Tran:2021ukl,Steinacker:2022jjv}.
Both the Segal construction and the HS-IKKT model \cite{Sperling:2017dts,Fredenhagen:2021bnw,Steinacker:2022jjv} are examples
of non-commutative field theories. Lorentz invariance is violated in generic
non-commutative theories due to explicit dependence of Poisson structure
$\theta^{\mu\nu}(x)$ on $x$. In the Segal construction, it is the phase-space
that is quantized and there is no explicit violation of Lorentz symmetry
($\theta^{\mu\nu}(x)$ pairs up $x$-$p$ and vanishes for $x$-$x$). In addition,
the spacetime Lagrangian is obtained via tracing out or averaging over
the $p$-fiber, which does not violate Lorentz symmetry (in fact, the averaging
over $p$, as we showed, can be performed in a general covariant manner).
In the HS-IKKT model, the trick is in having a nontrivial fibration
over the spacetime that is again averaged over without having
to violate Lorentz symmetry.

%%%%%%%%%%%%%%%%%%%%%%%%%%%%%%%%%%%%%%%%%%%%%%%%%%%%%%%%%%%%%
\section*{Acknowledgement}
%%%%%%%%%%%%%%%%%%%%%%%%%%%%%%%%%%%%%%%%%%%%%%%%%%%%%%%%%%%%%
We are grateful to Xavier Bekaert, Andreas Cap, Euihun Joung, Alexey Sharapov and Arkady Tseytlin for the very useful discussions and comments. E.S. is also grateful to Franz Cicery and Bernard de Wit for the very useful comments. T.B. is also grateful to Kevin Morand for enlightening discussions. This research was partially completed at the workshop ``Higher Spin Gravity and its Application'' supported by the Asia Pacific Center for Theoretical Physics. The work of T.B. and E.S. was partially supported by the European Research Council (ERC) under the European Union’s Horizon 2020 research and innovation programme (grant agreement No 101002551) and by the Fonds de la Recherche Scientifique --- FNRS under Grant No. F.4544.21. The work of T.B. was also supported by the European Union’s Horizon 2020 research and innovation program under the Marie Sk\l{}odowska Curie grant agreement No 101034383.

\appendix
%%%%%%%%%%%%%%%%%%%%%%%%%%%%%%%%%%%%%%%%%%%%%%%%%%%%%%%%%%%%%%
\section{BRST form of Segal}
\label{app:BRST}
%%%%%%%%%%%%%%%%%%%%%%%%%%%%%%%%%%%%%%%%%%%%%%%%%%%%%%%%%%%%%%

Consider the BFV phase space with coordinates $(x^a,p_b)$
of ghost number $0$, and a ghost pair $(c,b)$,
i.e. of respective ghost number $+1$ and $-1$. 
We use the language of symbols and star product
(for the moment we assume Moyal--Weyl star product). 
The nonvanishing star-commutators between
these coordinates are
\begin{equation}
    \qcommut{x^a}{p_b}=\hbar\delta^a_b\,,
    \qquad
    \qcommut{c}{b}=\hbar\,,
\end{equation}
and the algebra of functions in these coordinates
is equipped with an anti-involution $\dag$ defined by
\begin{equation}
   \hbar^\dagger=-\hbar,
   \quad
   x^\dagger=x\,,
   \quad
   p^\dagger=p\,,
   \quad
   c^\dagger=c\,,
   \quad
   b^\dagger=b\,,
\end{equation}
which verifies $(AB)^\dagger=(-)^{|A||B|}B^\dagger A^\dagger$
where $\lvert\cdot\rvert$ denotes the ghost number.
A generic nilpotent Hermitian BFV charge has the form
\begin{equation}
    \Omega=c\,F(x,p)\,,
    \quad
    \Omega^\dagger=\Omega\,,
\end{equation}
with $F^\dagger=F$ also Hermitian. Now we view $\Omega$
as a generating  function of fields, and consider
the following gauge theory:
\begin{equation}
    \qcommut {\Omega}{\Omega}=0\,,
    \quad
    \delta_\Xi{\Omega}=\frac{1}{\hbar}\qcommut{\Omega}{\Xi}\,,
    \qquad
    \gh{\Xi}=0\,,
    \quad
    \Xi^\dagger=\Xi\,,
\end{equation}
where $\Xi$ is a generating function of gauge parameters.
In our case $\Xi=\xi(x,p) + cb\,w(x,p)$ with $\xi^\dagger=\xi$ 
and $w^\dagger=-w$. In terms of components the gauge 
transformations read as:
\begin{equation}
    \delta_{\xi,w} F = \frac{1}{\hbar}\qcommut{F}{\xi}
    + \scommut{F}{w}\,.
\end{equation}
In particular $F$ remains Hermitian. It is easy to see
that the above precisely encode the Segal's gauge 
transformations.  It follows that Segal's gauge
transformations are precisely the natural symmetries
of the constrained Hamiltonian systems describing
the particle model. In particular, HS diffeomorphisms 
correspond to a quantized version of the canonical 
transformations of the constrained surface,
while HS Weyl transformations correspond to 
redefinitions of the constraint. 

The above BFV--BRST interpretation of the off-shell Segal 
system was proposed in~\cite{Grigoriev:2006tt}
(see also~\cite{Bekaert:2013zya,Grigoriev:2021bes}).
It is a useful starting point to construct the respective
parent reformulation~\cite{Grigoriev:2006tt}, from which
the parent Segal system is obtained by gauge-fixing
the gauge fields associated to Weyl transformations.

%%%%%%%%%%%%%%%%%%%%%%%%%%%%%%%%%%%%%%%%%%%%
\section{More on the star-Heaviside function}
\label{app:starHeaviside}
%%%%%%%%%%%%%%%%%%%%%%%%%%%%%%%%%%%%%%%%%%%%
In this appendix, we recall how to evaluate the Heaviside
star-function, as introduced and explained by Segal
\cite{Segal:2002gd}. Let us start with the definition
of a star-function: given a usual function which admits
an integral representation of the form
\begin{equation}
    f(x) = \oint_C dt\,\tilde f(t)\,e^{t\,x}\,,
\end{equation}
where $C$ is some contour in the complex plane, and
$\tilde f(t)$ a given function, the corresponding
star-function will be defined as
\begin{equation}
    f_\ast(a) = \oint_C dt\,\tilde f(t)\,e_\ast^{t\,a}\,,
\end{equation}
for any elements $a \in \cA_{2n}$ of the Weyl algebra
and where
\begin{equation}
    e_\ast^{t\,a} = \sum_{k=0}^\infty\,
    \tfrac{t^k}{k!}\,a^{\ast k}\,,
    \qquad 
    t \in \C\,,
    \qquad 
    a^{\ast k} := \underbrace{a \ast \dots \ast a}_{k\,\text{times}}\,,
\end{equation}
is the star-exponential. Any such star function can be 
expanded as a formal power series in $\hbar$, of the form
\cite[Sec. 5.2]{Segal:2002gd}
\begin{eqnarray}\label{eq:expansion_star-function}
    f_\ast(a) = \sum_{n=0}^\infty \hbar^{2n}\,
    \sum_{k=2}^{2n}\,f^{(k)}(a)\,p_{n,k}(a)
\end{eqnarray}
where $f^{(k)}$ denotes the $k$th derivative of $f$ and
$p_{n,k}(a)$ are monomials of $k$ in the first
$4n$ derivatives of $a$ with respect to the variables
of the Weyl algebra. In particular, the Heaviside
star-function is given in terms of the integral representation
\begin{equation}
    \Theta_\ast(a) := \lim_{\epsilon\to0^+}\,\tfrac1{2i\pi}\,
    \int_{-\infty}^{+\infty}\,
    \tfrac{d\tau}{\tau-i\epsilon}\,e^{i\tau\,a}_\ast\,,
\end{equation}
and admits a similar expansion in $\hbar$.

The crucial property of the Heaviside star-function is
that, according to the expansion \eqref{eq:expansion_star-function},
it verifies
\begin{equation} \label{quantumdelta}
    \Theta'_\ast(a) \ast a = \delta_\ast(a) \ast a = 0\,,
\end{equation}
since the derivative of the Heaviside distribution is the Dirac
distribution. Note that the above identity should be understood
in the sense of distributions. Fortunately, this is enough
to prove the invariance under higher spin Weyl transformations
of the Segal action \eqref{segal-action}, since its variation reads
\begin{equation}
    \delta_w S[F] = 2\,\int d^nx\,d^np\,\,\delta_\ast(F) \ast F \ast w
    = 2\,\int d^nx\,d^np\,\big(\delta_\ast(F) \ast F\big)\, w=0\,,
\end{equation}
where the second equality is obtained upon disregarding 
a total derivative.

Since \eqref{quantumdelta} can raise some doubts, let us illustrate
that the quantum identity $\delta_\ast(a) \ast a = 0$ reduces to
the classical one $\delta(a) a=0$ and its derivatives. Indeed,
let us start with the general expansion of
$f_\ast(H)$:\footnote{In this Appendix $\{Y^A\}$, with $A=1,\dots,2n$,
collectively denotes the $2n$ variables of the Weyl algebra
$\cA_{2n}$, $\pl_A\equiv \pl/\pl Y^A$ and
$H_{AB}\equiv \pl_A \pl_B H$.}
\begin{align}
    f_\ast(H)&= f(H) + \hbar^2\left(\tfrac{1}{6} H_{AB}H^A H^B f'''(H) +\tfrac14 H_{AB} H^{AB} f''(H)\right)+\mathcal{O}(\hbar^4)\,.
\end{align}
Now, we can write down the expansion of $f_\ast(H)\ast H$ to the same order:
\begin{align}\begin{aligned}
    f_\ast(H)\ast H &= f(H)H +\hbar^2\left( \tfrac{1}{6} H_{AB}H^A H^B f'''(H) +\tfrac14 H_{AB} H^{AB} f''(H)\right)H+\\
    &+\hbar^2\left(\tfrac{1}{2} H_{AB}H^A H^B f''(H) +\tfrac12 H_{AB} H^{AB} f'(H)\right) +\mathcal{O}(\hbar^4)\,. 
\end{aligned}
\end{align}
Of course, this is just an expansion of $g_\ast(H)$, where $g(H)=f(H)H$. In particular, we can rearrange it with the help of $g'=f'H+f$, $g''=f''H+2f'$, etc., to find
\begin{align}
    f_\ast(H)\ast H&= g(H) +\hbar^2\left( \tfrac{1}{6} H_{AB}H^A H^B g'''(H) +\tfrac14 H_{AB} H^{AB} g''(H)\right)+\mathcal{O}(\hbar^4)\,.
\end{align}
Lastly, we take $f(H)=\delta(H)$ and observe that the leading term is just the classical identity $\delta(H)H=0$, while the subleading ones are derivatives of it. Obviously, all of this is a consequence of the fact that the map $\rho: f(H)\mapsto f_\ast(H)$ is a homomorphism from the subalgebra of $\cC^\infty(T^*\manX)$ generated by $H$ to the  $\ast$-product subalgebra of $W(\manX)$ generated by $H$. Therefore, $\rho$ maps $f(H)g(H)$ to $f_\ast (H)\ast g_\ast(H)$, which we apply to $f=\delta(H)$ and $g=H$.

%%%%%%%%%%%%%%%%%%%%%%%%%%%%%%%%%%%%
\section{Modification of FFS cocycle}
\label{app:modification_FFS}
%%%%%%%%%%%%%%%%%%%%%%%%%%%%%%%%%%%%
Let us denote by $\Phi$ the Chevalley--Eilenberg cocycle
associated with the FFS cocycle, whose expression
is detailed below. For the sake
of conciseness, we will rewrite it as
\begin{multline}
    \Phi(a_0;a_1,\dots,a_{2n})  = \\  \int_{u\in\Delta_{2n}}
    \Big[\mathscr{D}(\partial_{y_0},\partial_{p_0},
    \partial_{y_i},\partial_{p_i},u)\, a_0(y_0,p_0)\,
    a_1(y_1,p_1)\,\dots\,a_{2n}(y_{2n},p_{2n})\Big]\Big|_{y_0=y_i=p_0=p_i=0}\,,
\end{multline}
where $a_0,a_1,\dots,a_{2n} \in \cA_{2n}$ are elements
of the Weyl algebra, $\Delta_{2n}$ is the standard $2n$-simplex
which can be defined as
\begin{equation}
    \Delta_{2n} = \big\{(u_1,\dots,u_{2n}) \in [0,1]^{2n}\,\rvert\,
    0 \leq u_1 \leq u_2 \leq \dots \leq u_{2n} \leq 1\big\}\,,
\end{equation}
and $\mathscr{D}(\partial_{y_0},\partial_{p_0},\partial_{y_i},\partial_{p_i},u)$ is
a function of the partial derivatives with respect
to the Weyl algebra variables $\{y_0^a, p_{a0},y^a_i,p_{ai}\}$, 
with $i=0,\dots,2n$ and $a=1,\dots,n$, and the simplex
coordinates $u$. Explicitly, it is given by
\begin{equation}
    \mathscr{D}(\partial_{y_0},\partial_{p_0},
    \partial_{y_i},\partial_{p_i},u)
    = \exp\Big[\hbar\sum_{0 \leq k<l \leq 2n}
    (\tfrac12+u_k-u_l)\,(\partial_{y_k} \cdot \partial_{p_l}
    -\partial_{p_k} \cdot \partial_{y_l})\Big]\,
    \det(\partial_{y_i},\partial_{p_i})\,,
\end{equation}
where $\partial_{y_i} \cdot \partial_{p_j}
= \tfrac{\partial}{\partial y^a_i}\,\tfrac{\partial}{\partial p_{aj}}$,
and $u_0=0$ by convention. Note 
that the determinant part of this operator {\it does not}
acts on the zeroth argument ($a_0$) of $\Phi$.

From the above cocycle, we can define a new one,
simply by not setting the $p_{ai}$ variables to zero
in the above expression but to the same value $p_a'$
for all $i=0,\dots,2n$, i.e.
\begin{multline}
    \tilde\Phi(a_0;a_1,\dots,a_{2n})(p') = \\ \int_{u\in\Delta_{2n}}
    \Big[\mathscr{D}(\partial_{y_0},\partial_{p_0},
    \partial_{y_i},\partial_{p_i},u)\,
    a_0(y_0,p_0)\, a_1(y_1,p_1)\,\dots\,
    a_{2n}(y_{2n},p_{2n})\Big]\Big|_{y_0=y_i=0,p_0=p_i=p'}\,,
\end{multline}
for any $a_0,a_1,\dots,a_{2n} \in \cA_{2n}$.
\begin{lemma}
    The map $\tilde\Phi$ defined above is a Chevalley--Eilenberg
    cocycle of degree $2n$ for the Lie algebra $(\cA_{2n},[-,-]_\ast)$
    associated with the Weyl algebra, with values in its dual
    whose coefficients are extended to the algebra of polynomials
    $\C[p_a']$ in $n$ variables.
\end{lemma}
\begin{proof}
    This simply follows from the fact that $\tilde\Phi$
    is obtained by pre-composing each one of the arguments
    of $\Phi$ by an automorphism of the Weyl algebra $\cA_{2n}$.
    Indeed, $\tilde\Phi$ is simply obtained from $\Phi$
    by shifting its argument by a parameter $p'$,
    i.e. $a_i(y_i,p_i) \to a_i(y_i,p_i+p')$.
    In other words,
    \begin{equation}
        \tilde\Phi(a_0;a_1,\dots,a_{2n})(p')
        = \Phi\big(T_{p'}(a_0);T_{p'}(a_1),\dots,
        T_{p'}(a_{2n})\big)\,,
    \end{equation}
    where we introduced the operator $T_{p'}(a)(y,p) := a(y,p+p')$,
    for any $a \in \cA_{2n}$. This operator is an automorphism
    of the Weyl algebra: the Moyal--Weyl star-product
    is invariant under $Sp_{2n} \ltimes \R^{2n}$,
    the semi-direct product of the symplectic group with
    the abelian group of translation in $2n$-dimensions,
    and $T_{p'}$ is nothing but the operator representing
    the abelian subgroup of $n$-dimensional translations.
    Since $\tilde\Phi$ is simply the composition
    of the cocycle $\Phi$ with automorphisms of the Weyl algebra,
    it follows directly that $\tilde\Phi$ verifies
    the same cocycle condition as $\Phi$, hence the lemma.
\end{proof}

Next, one can define a multilinear map
\begin{equation}
    \mu: \cA^{\,}_{2n} \otimes \cA_{2n}^{\wedge n} \to \C[p_a]\,,
\end{equation}
by the formula
\begin{equation}
    \mu(a_0|a_1, \dots, a_n)(p) := \tfrac1{n!}\,
    \epsilon_{b_1 \dots b_n}\,
    \tilde\Phi(a_0;a_1,\dots,a_d,y^{b_1},\dots,y^{b_n})(p)\,,
\end{equation}
for any $a_0,a_1,\dots,a_n \in \cA_{2n}$. Note that $\mu$
verifies\footnote{The map $\mu$ also vanishes identically
if at least one of its $n$ last arguments depends only
on the $y$ variables.}
\begin{equation}
    \tilde\Phi(f;p_{a_1},\dots,p_{a_n}) = \tfrac1{n!}\,
    \epsilon_{a_1 \dots a_n}\,f\rvert_{y=0}\,,
\end{equation}
for any $f \in \cA_{2n}$, as a consequence of the normalisation
condition \cite[Sec. 4.2, IV]{Feigin2005} of $\Phi$.
\begin{lemma}
    The map $\mu$ defined above verifies
    \begin{equation}\label{eq:almost_cocycle}
        \begin{aligned}
            \tfrac{\partial}{\partial p_a}\varphi_a(a_{-1}|a_0,\dots,a_n)
            & = \sum_{i=0}^n\,(-1)^i\,
            \mu([a_{-1},a_i]_\ast|a_0,\dots,\hat a_i,\dots,a_n) \\
            & \qquad + \sum_{i<j} (-1)^{i+j}\,\mu(a_{-1};[a_i,a_j]_\ast,
            a_0,\dots,\hat a_i,\dots,\hat a_j,\dots,a_n)\,,
        \end{aligned}
    \end{equation}
    with
    \begin{equation}
        \varphi_a(a_{-1}|a_0,\dots,a_n) := \tfrac{(-1)^{n-1}}{(n-1)!}\,
        \epsilon_{a b_1 \dots b_{n-1}}\,\tilde\Phi(a_{-1};a_0,\dots,a_n,
        y^{b_1},\dots,y^{b_{n-1}})\,,
    \end{equation}
    and for any $a_{-1},a_0,\dots,a_n \in \cA_{2n}$.
\end{lemma}
\begin{proof}
    This is a direct consequence of the fact that $\tilde\Phi$
    is a Chevalley--Eilenberg cocycle. Indeed, starting from
    \begin{equation}
      0=\tfrac1{n!}\,\epsilon_{b_1 \dots b_n}\,
      (\delta\tilde\Phi)(a_{-1};a_0,a_1,\dots,a_n,y^{b_1},\dots,y^{b_n})
    \end{equation}
    where $\delta$ denotes the Chevalley--Eilenberg differential,
    one finds
    \begin{equation}
      \begin{aligned}
        0 & = \sum_{i=0}^n\,(-1)^i\,
        \mu([a_{-1},a_i]_\ast|a_0,\dots,\hat a_i,\dots,a_n) \\
        & \qquad + \sum_{i<j} (-1)^{i+j}\,\mu(a_{-1};[a_i,a_j]_\ast,
        a_0,\dots,\hat a_i,\dots,\hat a_j,\dots,a_n) \\
        & \quad + \tfrac{(-1)^{n-1}}{(n-1)!}\,\epsilon_{ab_1 \dots b_{n-1}}\,
        \Big(\tilde\Phi([a_{-1},y^a]_\ast;a_0,a_1\dots,a_n,
        y^{b_1},\dots,y^{b_{n-1}}) \\
        & \qquad + \sum_{k=0}^n\,(-1)^k\,
        \tilde\Phi(a_{-1};[a_k,y^a]_\ast,a_0,\dots,\hat a_k, \dots, a_n,
        y^{b_1}, \dots, y^{b_{n-1}})\big)\,,
      \end{aligned}
    \end{equation}
    which, upon using $[-,y^a]_\ast = -\tfrac{\partial}{\partial p_a}$
    and the fact that the latter can be factored out of the expression
    of $\tilde\Phi$, reproduces \eqref{eq:almost_cocycle}.
\end{proof}
In plain words, $\mu$ is a Chevalley--Eilenberg cocycle,
up to a total derivative in $p$. Moreover, this property
implies that the extension of $\mu$ to forms on any manifold
taking values in the Weyl algebra can be used to define a trace
on the algebra of covariantly constant sections with respect to
any given flat connection $A$.
\begin{corollary}\label{cor:mu}
    Let $\manX$ be a smooth manifold, and $A$ a flat connection
    on its Weyl bundle, associated to $T^*\manX$. Then, for any pair of covariantly
    constant sections $F, G$ with respect to $A$,
    the $\Omega(\manX)$-linear extension of $\mu$ verifies
    the cyclicity condition
    \begin{equation}\label{eq:reduced_cyclicity}
        \mu([F,G]_\ast|A,\dots,A)\
        \propto\ n\hbar\,d\mu(F|G,A,\dots,A)
        + \tfrac{\partial}{\partial p_a}\varphi_a(F|G,A,\dots,A)\,,
    \end{equation}
    as well as the gauge invariance condition
    \begin{equation}\label{eq:reduced_gauge_invariance}
        \delta_\xi\mu(F|A,\dots,A)\
        \propto\ n\,d\mu(F|\xi,A,\dots,A)
        + \tfrac{\partial}{\partial p_a}\varphi_a(F|\xi,A,\dots,A)\,,
    \end{equation}
    for any gauge parameter $\xi$ (i.e. any section of
    the Weyl bundle, not necessarily covariantly constant).
\end{corollary}
\begin{proof}
    The proof is almost identical to that of \cite[Prop. 4.2]{Feigin2005}.
    For the sakes of completeness, let us repeat it.

    First, let us use the almost-cocycle condition
    \eqref{eq:almost_cocycle} verified by $\mu$ to write
    \begin{equation}
      \begin{aligned}
        \mu([F,G]_\ast|A,\dots,A)
        & = n\,\Big(\mu([F,A]_\ast|G,A,\dots,A) + \mu(F|[G,A]_\ast,A,\dots,A) \\
        & \qquad - \tfrac{(n-1)}2\,\mu(F|[A,A]_\ast,G,A,\dots,A)\Big)
        + \tfrac{\partial}{\partial p_a}\varphi_a(F|G,A,\dots,A)
      \end{aligned}
    \end{equation}
    where the dimension-dependent coefficients appear simply
    for combinatorial reason (several terms are identical
    since many of the arguments of $\mu$ are the same). Using
    the flatness of $A$ and covariant constancy of $F$ and $G$,
    we can replace all brackets by a differential term,
    \begin{equation}
      \begin{aligned}
        \mu([F,G]_\ast|A,\dots,A)
        & = n\hbar\,\Big(\mu(dF|G,A,\dots,A) + \mu(F|dG,A,\dots,A)\\
        & \qquad + (n-1)\,\mu(F|G,dA,\dots,A)\Big)
        + \tfrac{\partial}{\partial p_a}\varphi_a(F|G,A,\dots,A)\,,
      \end{aligned}
    \end{equation}
    which reproduces the exact term in \eqref{eq:reduced_cyclicity}.

    Similarly, the gauge variation
    \begin{equation}\label{eq:interm}
      \delta_\xi\mu(F|A,\dots,A)
      = \tfrac1\hbar\,\mu([F,\xi]_\ast|A,\dots,A)
      +  n\,\mu(F|d\xi+\tfrac1\hbar\,[A,\xi]_\ast,A,\dots,A)
    \end{equation}
    can be recast as
    \begin{equation}
      \begin{aligned}
        \delta_\xi\mu(F|A,\dots,A)
        & = n\,\Big(\mu(dF|\xi,A,\dots,A)
        + \tfrac1\hbar\,\mu(F|[\xi,A]_\ast,A,\dots,A) \\
        & \qquad  + (n-1)\,\mu(F|\xi,dA,\dots,A)
        + \mu(F|d\xi+\tfrac1\hbar\,[A,\xi]_\ast,A,\dots,A)\Big) \\
        & \qquad\qquad + \tfrac{\partial}{\partial p_a}\varphi_a(F|\xi,A,\dots,A)\,,
      \end{aligned}
    \end{equation}
    upon using the almost-cocycle condition on the first term
    of equation \eqref{eq:interm}, as well as the fact that $A$
    is flat and $F$ covariantly constant. The two terms containing
    a commutator $[A,\xi]_\ast$ cancel one another, and the remaining
    terms add up to give \eqref{eq:reduced_gauge_invariance}.
\end{proof}

%%%%%%%%%%%%%%%%%%%%%%%%%%%%%%%%%%%%%%%%%%%%%%%%%%%%%%%%%%%%%
\section{Quantization of the cotangent bundle}
\label{app:Fedosov}
%%%%%%%%%%%%%%%%%%%%%%%%%%%%%%%%%%%%%%%%%%%%%%%%%%%%%%%%%%%%%

%=================================================%
\paragraph{From off-shell Segal to the quantization
of the cotangent bundle.}
%=================================================%
The off-shell Segal system presented in Section
\bref{sec:parent_Segal} can be interpreted as a deformation
quantization of the cotangent bundle $T^*\manX$ over a given
manifold $\manX$, which we think of as spacetime here. Indeed,
recall that the algebra of functions on the cotangent bundle
which are polynomial in the momenta (fiber coordinates)
is isomorphic to the algebra of symbols $\cS(\manX)$
of differential operators on $\manX$. Finding such an isomorphism
(of vector spaces) between the space of symbols and
the space of differential operators $\cD(\manX)$ on $\manX$,
\begin{equation}
    \begin{tikzcd}
        \sigma: \cD(\manX)
        \ar[r, "\sim"]
        & \cS(\manX)\,,
    \end{tikzcd}
\end{equation}
amounts to a quantization of the cotangent bundle
(modulo some additional conditions, see e.g. \cite{Bekaert:2021sfc}
for more details and recent application in the context
of higher spin gravity) in the sense that it allows one
to define a star product on symbols, and hence on the algebra
of polynomial functions on $T^*\manX$, via the composition
of differential operators, i.e.
\begin{equation}
    \sigma(\hat D_1) \star \sigma(\hat D_2)
    = \sigma(\hat D_1 \circ \hat D_2)\,,
\end{equation}
for any differential operators $\hat D_1, \hat D_2 \in \cD(\manX)$.
For instance, in the case of flat space $\manX=\R^n$,
this is nothing but an ordering prescription (e.g. Weyl
or normal ordering). The connection $1$-form $A$
in the off-shell Segal system allows to define
such an isomorphism as follows: assuming that
$A$ is of the form
\begin{equation}
    A = dx^\mu\,e_\mu^a\,p_a + \dots
\end{equation}
where $x^\mu$ are coordinates on $\manX$, the dots denote terms
of higher orders in $y$ and $p$, and $e_\mu^a$ is invertible,
then there is a bijection between differential operators on $\manX$
and sections of the bundle of Weyl algebra (in the variables
$(y^a,p_a)$ with $a=1,\dots,\dim \manX=n$) over the manifold $\manX$,
which are annihilated by
\begin{equation}
    D := d + \tfrac1\hbar\,[A,-]_\ast\,,
\end{equation}
where $d$ denotes the de Rham differential on $\manX$ and $\ast$
the fiberwise Moyal--Weyl product (meaning, between the $(y^a,p_a)$
variables). Indeed, the equation $Df(x,p;y)=0$, can be solved
order by order in $y$ (and $\hbar$) thanks to the fact that
$e_\mu^a$ has an inverse (see Appendix \bref{app:proof}
for more details). In other words, the $y$-dependency
of any function annihilated by $D$ can be uniquely reconstructed,
so that
\begin{equation}
    Df = 0 
    \qquad \Leftrightarrow \qquad
    f = \tau(f_0)\,,
    \qquad
    f_0(x,p) = f(x,p;0)\,,
\end{equation}
where $\tau$ is a bijection (whose inverse consists in
setting $y=0$). Now on the one hand, functions depending
only on $x^\mu$ and $p_a$ are nothing but polynomial functions
on the cotangent bundle, expressed in a coordinate system
wherein the tautological (or Liouville) $1$-form
reads $\vartheta = dx^\mu\,e^a_\mu\,p_a$. On the other hand,
such functions identify with symbols of {\it fiberwise 
differential operators}, by which we mean differential 
operators in the $y$-variables, whose coefficients are
smooth functions on $\manX$ and formal series in $y$. As it turns out,
Dolgushev showed that the space of such symbols, which are
annihilated by the differential $D$, are in bijection with
differential operators on $\manX$ \cite[Th. 3 and Prop. 1]{Dolgushev:2003fg}.
In other words, the connection $1$-form in the off-shell
Segal system allows us to obtain a quantization of
the cotangent bundle, by establishing an isomorphism between
functions on $T^*\manX$ (which are polynomial in momenta),
and differential operators.

Considering that the cotangent bundle $T^*\manX$ of any manifold $\manX$
is symplectic, one can also quantize it using Fedosov's method
\cite{Fedosov:1994zz}. As it turns out, the question of 
the precise relation between these two approaches was studied
by Fedosov himself in \cite{Fedosov2001}. The result is that,
given a quantization of the cotangent bundle in the sense described
in the previous paragraphs, one can build a Fedosov quantization
of the cotangent bundle $T^*\manX$ as follows.
\begin{itemize}
\item First, one should define a Fedosov connection
on the Weyl algebra bundle over the symplectic manifold $T^*\manX$. It therefore defines a differential
on the space of forms valued in $\cA_{2n}$
that should take the form
\begin{equation}
    \tilde D = \tilde d + [\tilde A,-]_\ast\,,
    \label{eq:tilde_Fedosov}
\end{equation}
where $\tilde d$ denotes the de Rham differential
on $T^*\manX$ and $\tilde A$ is a $\cA_{2n}$-valued
connection $1$-form on $T^*\manX$, whose pieces linear in $y$
and $p$ read
\begin{equation}
    \tilde A = -d\pi_a\,y^a + dx^\mu\,e_\mu^a\,p_a
    + \dots\,,
\end{equation}
in a patch with coordinates $(x^\mu,\pi_a)$ of the cotangent
bundle. Such a $1$-form can be constructed from the $1$-form
in the parent Segal system, via
\begin{equation}
    \tilde A(x,\pi;y,p) = -d\pi_a\,y^a + A'\,,
    \qquad
    A' := A(x;y,p+\pi)\,.
\end{equation}
The expression \eqref{eq:tilde_Fedosov} does square to zero
since
\begin{equation}
    \tilde d \tilde A + \tfrac12\,[\tilde A, \tilde A]_\ast
    = 0
    \qquad \Leftrightarrow \qquad
    \left\{
    \begin{aligned}
        0 & = (\tfrac{\partial}{\partial \pi_a}
        - \tfrac{\partial}{\partial p_a})\,A_\mu'\,,\\
        0 & = dA' + \tfrac12\,[A',A']_\ast\,.
    \end{aligned}
    \right.
\end{equation}
The first equation is satisfied as a consequence of
the particular dependency of $A'$ in $p$ and $\pi$,
and the second one by virtue of the fact that both
the action of the de Rham differential and the Moyal--Weyl
star product commute with translation in $p$\,\footnote{In fact,
one can characterize the Moyal--Weyl star product
as the unique star product on $\R^{2n}$ which is 
$\mathfrak{isp}(2n,\R)$-equivariant \cite{Gutt1983}.},
and therefore this second equation follows from flatness
of $A$. This is also consistent with the fact that
the Fedosov connection $\tilde A$ leads to a quantization
of the cotangent bundle equipped with its canonical
symplectic form $\tilde d(e^a\,\pi_a)$. Indeed, adding
the central term $-e^a\,\pi_a$ to $\tilde A$ results in
the minimal Fedosov connection whose curvature is
$\tilde d(e^a\,\pi_a)$, in accordance with the standard
Fedosov quantization.

\item Second, one should find covariantly constant sections 
of this differential $\tilde D$, i.e. $0$-forms $\tilde F$
such that
\begin{equation}
    \tilde D \tilde F = 0
    \qquad \Leftrightarrow \qquad
    \left\{
        \begin{aligned}
            0 & = \big(\tfrac{\partial}{\partial \pi_a}
            - \tfrac{\partial}{\partial p_a}\big)\,\tilde F\,,\\
            0 & = d \tilde F + [A', \tilde F]_\ast\,,
        \end{aligned}
    \right.
\end{equation}
where the above equation has been split into two according to
the natural basis $(dx^\mu, d\pi_a)$ of one-forms on $T^*\manX$.
The first equation is simplify solved by
\begin{equation}
    \tilde F(x,\pi;y,p) = F(x;y,p+\pi)\,,
\end{equation}
and hence the second equation becomes equivalent to $DF = 0$.
\end{itemize}
In summary, the parent Segal system can be extended
into a Fedosov differential and a covariantly constant
section of the Weyl bundle over $T^*\manX$.

%===========================================%
\paragraph{Traces in Fedosov's quantization.}
%===========================================%
As previously recalled, the basic principle of Fedosov's
quantization is to establish an isomorphism, usually
denoted $\tau$, between functions on a symplectic manifold
$(M,\omega)$ and sections of the Weyl bundle which
are annihilated by the Fedosov differential
$D_M := d_M + [A_M,-]_\ast$, where $d_M$ denotes the de Rham
differential on $M$, so as to define the star product $\star$
on $\cC^\infty(M)$ as the pullback of the fiberwise
Moyal--Weyl product $\ast$ of the corresponding sections
of the Weyl bundle (under the action of $\tau$).
In other words,
\begin{equation}
    \begin{tikzcd}
        \tau:\,\big(\cC^\infty(M), \star\big)
        \ar[r, "\sim"]
        & \big(H^0(D_M),\ast\big)\,,
    \end{tikzcd}
\end{equation}
is an isomorphism of associative algebras. This allows
one to define a trace on $\big(\cC^\infty(M),\star\big)$
by using structures defined on the Weyl bundle, namely
the Fedosov differential and a Hochschild cocycle
of the Weyl algebra. More precisely, given a function
$f \in \cC^\infty(M)$, its trace is given by
\begin{equation}
    \Tr_{A_M}(f)
    := \int_M \Phi(\tau(f);A_M, \dots, A_M)\,,
    \label{eq:def_trace}
\end{equation}
where
\begin{equation}
    \Phi: \underbrace{\cA_{2n} \otimes \dots
    \otimes \cA_{2n}}_{2n\,\text{times}} \to \cA_{2n}^*\,,
\end{equation}
the representative of the Hochschild cohomology of $\cA_{2n}$
with values in its dual $\cA_{2n}^*$ derived in \cite{Feigin2005}.
That this (multi-linear) map is a cocycle means that it verifies
\begin{equation}
    0=\Phi(f \ast a_0; a_1, \dots, a_{2n})
    + \sum_{k=1}^{2n} (-1)^k\, \Phi( f; a_0, \dots,
    a_{k-1} \ast a_k, \dots, a_{2n})
    - \Phi(a_{2n} \ast f; a_0, \dots, a_{2n-1})\,,
\end{equation}
for any $a_0, \dots, a_{2n}, f \in \cA_{2n}$. Antisymmetrizing
this identity in the $a_i$ arguments yields\,\footnote{Note that
this operation corresponds to producing a Chevalley--Eilenberg
cocycle for the commutator algebra out of the original Hochschild
cocycle.}
\begin{equation}
    0 = \sum_{\sigma\in\cS_{2n+1}} (-1)^\sigma\,
    \Big(\Phi([f,a_{\sigma_0}]_\ast;a_{\sigma_1},\dots,a_{\sigma_{2n}})
    - \tfrac12\,
    \Phi(f;[a_{\sigma_0},a_{\sigma_1}]_\ast,\dots,a_{\sigma_{2n}})\Big)\,,
\end{equation}
which can be re-written as
\begin{eqnarray}
    0 & = & \sum_{i=0}^{2n}\,(-1)^i\,
    \Phi([f,a_i]_\ast;a_0,\dots,\hat a_i,\dots,a_{2n}) \\
    && \qquad + \sum_{i<j} (-1)^{i+j}\,\Phi(f;[a_i,a_j]_\ast,a_0,\dots,
    \hat a_i,\dots,\hat a_j,\dots,a_{2n})\,,
\end{eqnarray}
where hats denote the omission of the arguments, and with
the understanding that the $\Phi$ appearing in this formula
is the antisymmetric part in the last $2n$ arguments of the original
cocycle. Using the above identity, and the fact that
\begin{equation}
    d_M A_M + \tfrac12\,[A_M, A_M]_\ast = 0\,,
    \qquad 
    D_M F = 0 = D_M G\,,
\end{equation}
for $F = \tau(f)$ and $G = \tau(g)$ the lifts of a pair of
functions $f, g \in \cC^\infty(M)$,
one can check that \cite[Prop. 4.2]{Feigin2005}
\begin{equation}
    \Phi\big([F, G]_\ast; A_M, \dots, A_M\big)
    = d_M\,\Phi\big(F;G, A_M, \dots, A_M\big)\,,
\end{equation}
and 
\begin{equation}
    \delta_\xi \Phi(F;A_M, \dots, A_M)
    = d_M\,\Phi\big(F;\xi, A_M, \dots, A_M\big)\,,
\end{equation}
for any section $\xi$ of the Weyl bundle. In plain words,
the obstruction for $\Phi(\tau(-);A_M,\dots,A_M)$ to vanish
on $\star$-commutator and to be gauge-invariant is exact,
which implies that for compactly supported functions,
the operation \eqref{eq:def_trace} does define a trace,
and that this definition is independent of the choice
of Fedosov differential $D_M$.

%%%%%%%%%%%%%%%%%%%%%%%%%%%%%%%%%%%%%%%%%%%%%
\section{Flat connection in the Weyl algebra}
\label{app:proof}
%%%%%%%%%%%%%%%%%%%%%%%%%%%%%%%%%%%%%%%%%%%%%
In this Appendix, we give a proof of Proposition \ref{prop:flattening},
as well as some details about the construction
of covariantly constant sections. In the rest
of this section, $\deg(-)$ will refer to the Fedosov
degree \eqref{Fedosov-degree} on the Weyl algebra.
\begin{proof}[Proof of Prop. \ref{prop:flattening}]
    Let the Weyl bundle connection be given by
    $\cD=d+\tfrac1\hbar\,[\varpi,-]_\ast$, where
    \begin{equation}
        \varpi = dx^\mu\,e_\mu^a\,p_a + \omega\,,
        \qquad \text{with} \qquad
        \omega := \sum_{n\geq2} \varpi_{(n)}\,,
        \qquad
        \deg(\varpi_{(n)})=n\,,
    \end{equation}
    with $e^a_\mu$ invertible, and decompose its curvature
    \begin{equation}
        R := d\varpi + \tfrac1{2\hbar}\,[\varpi,\varpi]_\ast
        = \sum_{n\geq1} R_{(n)}\,,
    \end{equation}
    according to the Fedosov degree. Let us furthermore consider
    \begin{equation}
        A = \varpi + W\,,
        \qquad \text{with} \qquad 
        W = \sum_{n\geq2} W_{(n)}\,,
        \qquad
        \deg(W_{(n)}) = n\,,
    \end{equation}
    whose curvature is given by
    \begin{equation}
        dA + \tfrac1{2\hbar}\,[A,A]_\ast
        = R - \delta W + \cD_\omega W
        + \tfrac1{2\hbar}\,[W,W]_\ast\,,
    \end{equation}
    where 
    \begin{equation}
        \delta := -\tfrac1\hbar\,[e^a\,p_a,-]_\ast
        = e^a\,\tfrac{\partial}{\partial y^a}\,,
        \qquad
        \cD_\omega := \tfrac1\hbar\,[\omega,-]_\ast\,,
    \end{equation}
    are derivations of degree $\deg(\delta)=-1$
    and $\deg(\cD_\omega) \geq 0$ (by which we mean
    that $\cD_\omega$ maps elements of degree $k$
    to elements of degree higher or equal to $k$).
    In fact, $\delta$ defines a differential
    on $\Omega(\manX,\hat \cA_{2n})$ since $\delta^2=0$.
    At order $n$, the flatness condition for $A$ reads 
    \begin{equation}\label{eq:flatness}
        \delta W_{(n+1)} = R_{(n)} + (\cD_\omega W)_{(n)}
        + \tfrac1{2\hbar}\,\sum_{k=2}^n [W_{(k)},W_{(n+2-k)}]_\ast\,,
    \end{equation}
    and in particular, the right hand side of the above
    equation only contains the components $W_{(k)}$ with $k \leq n$.
    Since the right hand side is $\delta$-closed by virtue
    of the fact that $(-\delta+\cD_\omega)R = 0$,
    which corresponds to the Bianchi identity of
    the connection $\cD$, and $[W,[W,W]_\ast]_\ast=0$
    (Jacobi identity), the above equation can be read
    as the condition that the right hand side is $\delta$-exact.
    In other words, assuming that there exists $W_{(k)}$
    with $1 \leq k \leq n$ such that the curvature of $A$
    vanishes in degrees lower or equal to $n-1$,
    then the existence of $W_{(n+1)}$ ensuring the vanishing
    of the curvature of $A$ up to degree $n$ is encoded
    in the cohomology group of $\delta$ in (form) degree $1$.
    The contracting homotopy $h$, introduced
    in \eqref{eq:contracting}, can therefore
    be used to solve $W_{(n+1)}$ in terms of the lower
    components, namely
    \begin{equation}\label{eq:recursion_connection}
        W_{(n+1)} = h\Big(R_{(n)} + (\cD_\omega W)_{(n)} 
        + \tfrac1{2\hbar}\,\sum_{k=2}^n [W_{(k)},W_{(n+2-k)}]_\ast\Big)\,,
    \end{equation}
    which solves \eqref{eq:flatness} by virtue of the fact
    that $\{h,\delta\}=1$ on $\cA_{2n}$-valued $1$-forms
    and $h^2=0$, and therefore uniquely specifies $W$
    in terms of the curvature of $\cD$.

    Having solved for $A$, we can now solve for $F$
    in a similar manner. Indeed, the covariant constancy
    condition reads
    \begin{equation}
        \delta F_{(n+1)} = (\cD_\omega F)_{(n)}
        + \tfrac1\hbar\,
        \sum_{k=2}^{n+2}\,[W_{(k)},F_{(n+2-k)}]_\ast = 0\,,
    \end{equation}
    at order $n$. Here again, one can use the contracting
    homotopy to express $F_{(n+1)}$ in terms of lower order
    components, namely
    \begin{equation}
        F_{(n+1)} = h\Big((\cD_\omega F)_{(n)} + \tfrac1\hbar\,
        \sum_{k=2}^{n+2}\,[W_{(k)},F_{(n+2-k)}]_\ast\Big)\,.
    \end{equation}
    Note that $h(F)=0$ due to the fact that $F$ is a $0$-form.
\end{proof}

This way of constructing
a flat connection starting from a possibly curved one
is a variation on the Fedosov approach to the deformation
quantization problem of symplectic manifold
\cite{Fedosov:1994zz}, as explained in
\cite[App. A]{Grigoriev:2016bzl} (see also 
\cite{Dolgushev:2003fg}). By construction, the completion of
some possibly curved connection $\cD$ into a flat
connection $A$ will only contain the curvature of $\cD$,
as well as covariant derivatives and contraction thereof.

Note that in the previous proof, we implicitly assumed that
the connection $\cD$ has no degree-$0$ piece, i.e.
that $\varpi$ does not contains a term $a \in \Omega^1(\manX)$,
which lies in the center of the Weyl algebra. If such a term 
is present, Proposition \ref{prop:flattening} still holds
due to the following line of argument. Since $a$ has degree $0$,
its field strength $da$ contributes to the curvature
of the connection $\varpi=a + e^a\,p_a + \omega$ in degree $0$
as well and hence requires the introduction of a term $W_{(1)}$
of degree $1$ in $W$ to be compensated, i.e.
$\delta W_{(1)}=da$. This equation is solved by
$W_{(1)}=-e^a\,\partial_a a_b\,y^b$, and the obstruction
to the flatness of the connection $\varpi+W_{(1)}$
is now of degree $2$ and higher, and the degree $1$ piece
of this connection defines a new differential
\begin{equation}
    \delta' := -\tfrac1\hbar\,e^a\,
    [\,p_a-\partial_a a_b\,y^b,-]_\ast\,,
\end{equation}
with respect to which the condition that $A=\varpi+W$
is flat reads as in \eqref{eq:flatness} upon replacing
$\delta$ with $\delta'$. Consequently, the existence of 
$W$ also boils down to a cohomology problem, i.e. it is
guaranteed provided that the cohomology of $\delta'$
is empty in form degree $1$. As a matter of fact,
one can show that the cohomology of $\delta'$ is empty
in form degree greater or equal to one.\footnote{Let us
start by looking for $\delta'$-cocycles, i.e. solutions
to $\delta'\alpha=0$. This equation can be expanded
with respect to the degree in $y$
as $\delta \alpha_{(k+1)}+\delta_a \alpha_{(k)}=0$,
where $\delta_a:=\tfrac1\hbar\,[e^a\partial_a a_b\,y^b,-]_\ast$.
Since $\delta$ and $\delta_a$ commute, $\delta_a\alpha_{(k)}$
is a $\delta$-cocycle and hence the existence of $\alpha_{(k+1)}$
is ensured whenever $\alpha$ has form (strictly) positive
form degree. This shows existence of $\delta'$-cocycles
in any positive form degree. Now given such a cocycle $\alpha$,
let us look for a solution to $\delta'\beta=\alpha$.
Once again, expanding this equation with respect to
the degree in $y$, one can see that it can be solved
iteratively by $\beta_{(k+1)}=h(\alpha_{(k)}+\delta_a\beta_{(k)})$.}
For the same reason, the lift of sections of $S(T\manX)$
to covariantly constant sections of the Weyl bundle $W(\manX)$
also exists.

Let us note that elements of the Weyl algebra which are
linear in $p$ form a Lie subalgebra under the star-commutator.
In fact, the commutator of these elements reduces to the Poisson
bracket, i.e.
\begin{equation}
    \tfrac1\hbar\,[f,g]_\ast
    = \big(\tfrac{\partial f}{\partial y^a}\,
    \tfrac{\partial g}{\partial p_a}
    - \tfrac{\partial f}{\partial p_a}\,
    \tfrac{\partial g}{\partial y^a}\big)
    =: \{f,g\}\,,
\end{equation}
for $f,g \in \cA_{2n}$ linear in $p$.\footnote{This subalgebra
is often referred to as the algebra of formal vector
fields. This algebra is a central piece of `formal geometry'
\cite{Gelfand:1971jg, Gelfand1972}, 
and can be obtained as the prolongation (in the sense
of Kobayashi \cite{Kobayashi2012}) of the inhomogeneous
general linear algebra $\mathfrak{igl}(n,\R)$.}
Consequently, the completion of a connection $\varpi$
which takes values in the subalgebra of elements linear in $p$
will also share this property, since the recursive procedure
\eqref{eq:recursion_connection} will only produce terms
taking values in this subalgebra.

Let us also remark that, if $\varpi$ is linear in $p$,
and hence so is $A$, then the highest power in $p$ contained
in covariantly constant section $F$ is the same as that
of $F_{(0)}$. More precisely, if $F_{(0)}$ is an homogeneous
polynomial in $p$ of degree $s$, then $F$ is a polynomial
in $p$ with terms of degrees $s$, $s-2$, $s-4$, \dots
\cite[Prop. A.2]{Grigoriev:2016bzl}.

%%%%%%%%%%%%%%%%%%%%%%%%%%%%%%%%%
\section{Flattening the HS frame}
\label{app:proof-HSframe}
%%%%%%%%%%%%%%%%%%%%%%%%%%%%%%%%%
Here we give an alternative procedure of encoding
the unconstrained CHS fields into the flat connection $A$.
In this context, it is convenient to employ an alternative degree:
\begin{equation}
    \deg(y) = 1\,, \qquad \deg(p)=0\,,
    \qquad 
    \deg(\hbar) = 1\,,
\end{equation}
and expand $A$ as $A=A_0+A_1+\ldots$ into the homogeneous components.
We treat $A_0=\sum_s dx^\mu E_\mu^{b(s)}p_{b(s)}$ as the initial data 
and, as before, assume that the term linear in $p$ in $A_0$
is invertible.  In this setup, the role of $\delta$ is played
by a new operator of degree $-1$:
\begin{equation}
\delta^\prime=-\frac{1}{\hbar}\qcommut{A_0}{\cdot}=
\delta+\sum_{l=1}^\infty \delta_l\,,
\end{equation}
where $\delta_l$ has the total homogeneity degree $l$ in $p_a$
and  $\hbar$. The standard homological algebra argument then shows
that $\delta^\prime$-cohomology is empty in the nonvanishing
form-degree. Moreover, the respective contracting homotopy operator
can be chosen to have homogeneity $1$ in $y^a$.\footnote{This happens
because $\delta$ has no cohomology in nonvanishing form-degree,
while the complex is quasi-isomorphic to $(H^\bullet(\delta),\Delta)$, 
where  $\Delta$ is a differential induces by $\delta^\prime$
in $H^\bullet(\delta)$. Note, however, that contracting homotopy implies inverting the HS frame and hence could result in elements which are nonpolynomial in $p_a$.} 

At degree $1$, we take $A_1=a_1+\Gamma$,
where $\Gamma=dx^\mu\omega_{\mu}{}^c{}_b\,y^b p_c$ encodes
coefficients of the bare affine connection. The role of $\Gamma$
is to maintain covariance of the procedure. Indeed,
as we are going to see, with this choice only covariant derivatives
of the fields entering $E_\mu(p)$ enter the construction.
At degree zero, the equation $dA+\tfrac1{2\hbar}\,[A,A]_\ast=0$
implies
\begin{equation}
    \delta^\prime a_1=\nabla A_0
    \qquad \Leftrightarrow \qquad
    dA_0+\tfrac1\hbar\,\qcommut{A_1}{A_0}=0
\end{equation}
and can be interpreted 
as the fact that the HS frame $A_0$ is covariantly constant
with respect to HS connection 
$A_1$. Because $A_0$ is $y$-independent, $\delta^\prime A_0=0$
so that the consistency condition $\delta^\prime\nabla A_0=0$ 
is fulfilled and hence $a_1$ exists. Such a ``connection''  
$A_1$ can be considered the ``torsion-free'' HS connection.
Note that to avoid the dependency on the arbitrary connection
$\Gamma$, one could take as $\Gamma$ a Levi--Civita connection 
determined by the spin-2 frame $e^a_\mu$ and the constant Minkowski 
metric $\eta^{ab}$.

At degree $1$, we have
$\delta^\prime A_2=d A_1+\tfrac1{2\hbar}\,[A_1,A_1]_\ast$,
The consistency condition again holds:
\begin{equation}
    \delta^\prime\left[d A_1+\ffrac{1}{2\hbar}\qcommut{A_1}{A_1}\right]
    =-d\delta^\prime A_1 +\tfrac1\hbar\qcommut{dA_0}{A_1}
    +\tfrac1\hbar\qcommut{\delta^\prime A_1}{A_1}=0\,,
\end{equation}
where we have made use of 
$\delta^\prime A_1=dA_0$, and hence $A_2$ also exists.

One then proceeds by the standard induction. Let us recall how it goes.
Introducing $\Omega=\hbar d+A$, the zero-curvature condition
is equivalent to  $\qcommut{\Omega}{\Omega}=0$. At degree $k$,
the equation reads
\begin{equation}
    \delta^\prime A_{k+1}
    +\ffrac{1}{2\hbar}\qcommut{\Omega^k}{\Omega^k}\rvert_k=0\,,
    \qquad \text{with} \qquad 
    \Omega^k := \hbar d+\sum_{l=0}^k A_l\,,
\end{equation}
where $C|_{k}$ denotes the degree $k$ component of $C$,
and it is assumed that the equation is already solved to order $k$.
The consistency condition
$\delta^\prime (\qcommut{\Omega^k}{\Omega^k}|_k)=0$
is the degree $k$ component of the identity
$\qcommut{\Omega^k}{\qcommut{\Omega^k}{\Omega^k}}=0$,
where the induction assumption has been used. It is therefore
satisfied, which implies the existence of $A_{k+1}$. In other words,
any $A_0=dx^\mu E_\mu(p)$  has a unique (modulo gauge transformations) 
lift to $A$ satisfying the flatness condition
$dA+\tfrac1{2\hbar}\,[A,A]_\ast=0$. It is natural to call
$dx^\mu E_\mu(p)$ a HS frame-field. 

\footnotesize
\setlength{\itemsep}{0em}
\providecommand{\href}[2]{#2}\begingroup\raggedright\endgroup

\end{document}